\PassOptionsToClass{cleveref}{jmlr}
\documentclass{colt2026}
\usepackage[
	paperwidth=8.5in,
    paperheight=11in,
    margin=1.25in,
    marginparwidth=0.75in,
    marginparsep=0.1in,
    footskip=0.25in
]{geometry}

\usepackage[shortlabels]{enumitem}
\setlist{
	itemsep=1ex,
	listparindent=\parindent,
	parsep=0em,
	topsep=2ex
}
\allowdisplaybreaks[1]

\usepackage{amsfonts}
\usepackage{amsmath}
\usepackage{amssymb}
\usepackage{bm}
\usepackage[mathscr]{euscript}
\usepackage{mathtools}

\usepackage{booktabs} 
\usepackage{comment}
\usepackage{float} 
\usepackage{graphicx}
\usepackage{xcolor}

\newcommand\linkcolor{black}
\newcommand\pref[1]{\textcolor{\linkcolor}{(\ref{#1})}}

\usepackage{hyperref}
\hypersetup{
    colorlinks,
    linktoc=all,
    citecolor=black,
    linkcolor=\linkcolor,
    urlcolor=\linkcolor
}

\makeatletter
\newcommand{\assumplabel}[2]{%
  \Hy@raisedlink{\hyper@anchorstart{assump.#1}\hyper@anchorend}%
  \protected@write\@auxout{}{%
    \string\newlabel{#1}{{#2}{\thepage}{}{assump.#1}{}}%
  }%
}
\makeatother



\numberwithin{theorem}{section} 
\aliascntresetthe{lemma}
\aliascntresetthe{proposition}
\aliascntresetthe{remark}
\aliascntresetthe{corollary}
\aliascntresetthe{definition}
\aliascntresetthe{conjecture}
\aliascntresetthe{axiom}

\makeatletter 
\renewenvironment{proof}[1][\proofname]{%
  \par\noindent{\bfseries\upshape #1\ }%
  \ignorespaces
}{\jmlrQED}
\makeatother

\crefname{fact}{Fact}{Facts}
\crefname{inequality}{Inequality}{Inequalities}
\creflabelformat{inequality}{#2(#1)#3}
\crefname{lemma}{Lemma}{Lemmas}
\crefname{section}{Section}{Sections}

\newcommand{\bbE}{{\mathbb E}}
\newcommand{\E}{\mathbb{E}}

\newcommand{\bbP}{{\mathbb P}}

\newcommand{\bbR}{{\mathbb R}}
\newcommand{\R}{\mathbb{R}}
\newcommand{\bbS}{{\mathbb S}}

\newcommand{\bsone}{\boldsymbol{1}}
\newcommand{\cB}{{\mathcal B}}

\newcommand{\cS}{{\mathcal S}}

\newcommand{\GIP}{G_{\mathrm{IP}}}

\newcommand{\bA}{\mathbf{A}}

\newcommand{\bb}{\mathbf{b}}
\newcommand{\bB}{\mathbf{B}}
\newcommand{\bd}{\mathbf{d}}
\newcommand{\bD}{\mathbf{D}}
\newcommand{\bG}{\mathbf{G}}
\newcommand{\bJ}{\mathbf{J}}
\newcommand{\bL}{\mathbf{L}}
\newcommand{\bM}{\mathbf{M}}
\newcommand{\bH}{\mathbf{H}}
\newcommand{\bK}{\mathbf{K}}
\newcommand{\bW}{\mathbf{W}}
\newcommand{\bS}{\mathbf{S}}
\newcommand{\bX}{\mathbf{X}}
\newcommand{\bY}{\mathbf{Y}}
\newcommand{\bZ}{\mathbf{Z}}
\newcommand{\bQ}{\mathbf{Q}}
\newcommand{\bR}{\mathbf{R}}
\newcommand{\bI}{\mathbf{I}}

\newcommand{\bLam}{\mathbf{\Lambda}}

\newcommand{\eps}{\varepsilon}

\newcommand{\ip}[1]{\left\langle #1 \right\rangle}

\usepackage{pifont}

\newcommand\erdos{Erd\H{o}s}
\newcommand\renyi{R\'enyi}
\newcommand\erdosrenyi{\erdos--\renyi}

\makeatletter
\newcommand{\vast}{\bBigg@{3}}
\makeatother

\DeclareMathOperator{\argmax}{argmax}

\DeclareMathOperator{\Ber}{Ber}
\DeclareMathOperator{\Cov}{Cov}
\DeclareMathOperator{\diag}{diag}

\DeclareMathOperator{\Unif}{Unif}
\DeclareMathOperator{\Var}{Var}

\DeclareMathOperator{\Tr}{Tr}
\newcommand\Corr{\mathsf{Corr}}
\newcommand\CorrleD{\Corr_{\leqslant D}}

\DeclarePairedDelimiter\inner{\langle}{\rangle}
\DeclarePairedDelimiter\abs{\lvert}{\rvert}
\DeclarePairedDelimiter\norm{\lVert}{\rVert}

\title{Spectral Recovery of a Planted Triangle-Dense Subgraph}
\usepackage{times}

\coltauthor{%
    \Name{Sam {van der Poel}} \Email{samvanderpoel@gatech.edu}\\
    \Name{Cheng Mao} \Email{cheng.mao@math.gatech.edu}\\
    \Name{Benjamin M\textsuperscript{c}Kenna} \Email{mckenna@math.gatech.edu}\\
    \addr School of Mathematics, Georgia Institute of Technology, Atlanta, GA%
}

\makeatletter
\newcommand{\blfootnote}[1]{%
  \begingroup
  \renewcommand{\thefootnote}{}%
  \renewcommand{\footnoteseptext}{}%
  \footnotetext{#1}%
  \endgroup
}
\makeatother

\begin{document}

\maketitle
\makeatletter
\def\@shortauthor{Van der Poel Mao M\textsuperscript{c}Kenna}
\makeatother

\begin{abstract}
Given a simple graph on $n$ vertices and a parameter $k$, the triangle-densest-$k$-subgraph problem is known to be computationally hard in the worst case. To circumvent the computational hardness, we study an average-case model where a triangle-dense subgraph on $k$ vertices is planted in an Erd\H{o}s--R\'enyi random graph on $n$ vertices. For the recovery of the planted subgraph, we propose a simple spectral algorithm and a semidefinite program, both of which use a graph matrix whose entries are local signed triangle counts. Theoretical guarantees for these algorithms are established through spectral analysis of the graph matrix. Finally, we provide evidence showing a statistical-to-computational gap analogous to that for the planted clique problem. The computational threshold in terms of the subgraph size $k$ is at least $\sqrt{n}$ in the framework of low-degree polynomial algorithms, while the information-theoretic threshold is at most logarithmic in $n$.
\end{abstract}


\blfootnote{Accepted for presentation at the Conference on Learning Theory (COLT) 2026.}


\section{Introduction}

Planted models have achieved significant success in the study of average-case computational hardness for inference on random graphs. A canonical example is the densest $k$-subgraph problem: although it is NP-hard in the worst case, its planted variants on random graphs---most notably the iconic planted clique problem \citep{jerrum1992large,alon1998finding} and the planted dense subgraph problem \citep{feige1997densest,bhaskara2010detecting}---have generated a substantial body of work on computational thresholds and the development of efficient algorithms. 

In this work, we propose a planted model for a triangle-dense subgraph in a random graph and study the recovery of the planted subgraph. 
Our first motivation comes from the \emph{triangle-densest $k$-subgraph problem}, which refers to finding a subset of $k$ vertices that induces the maximum number of triangles in a graph on $n$ vertices. This problem, as well as the approximation of the optimal value, is known to be computationally hard in the worst case~\citep{konar2022triangle}. 
Our work proposes an average-case model and efficient algorithms for this problem, analogous to the study of the planted dense subgraph model for the densest $k$-subgraph problem.

The model we propose plants a $k$-subgraph in an ambient \erdosrenyi{} graph on $n$ vertices, where the subgraph is a \emph{random geometric graph} (RGG) generated from a linear inner product kernel. RGGs are latent space models in which connections between nodes depend on the geometric proximity of corresponding node features in a latent space. The geometric structure of RGGs makes them useful in applications including social networks \citep{hoff2002latent}. Further, one of the hallmark properties of RGGs is their elevated \emph{triangle density}, which has been described as one of the defining characteristics of social networks \citep{easley2010networks,sala2010measurement,ugander2013subgraph,gupta2014decompositions}. 
This phenomenon has also played a significant role in recent theoretical studies of RGGs, as the \emph{signed triangle count} has been identified as the optimal statistic for distinguishing an RGG from an \erdosrenyi{} graph in various settings \citep{bubeck2016testing,liu2022testing,liu2023probabilistic,bangachev2025random,mao2026random}. 
From this perspective, our model is also an average-case model for planted geometry in a random graph.


\subsection{Main contributions}
We now briefly introduce our model and results. 
Suppose that we plant a subset $\cS$ of size $k$ in the vertex set $[n] := \{1,\dots,n\}$. The observed graph, identified with its adjacency matrix $\bA$, follows the \erdosrenyi{} model $G(n,p)$ for $p \in (0,1/2]$, except that the planted subgraph $\bA[\cS]$ induced by $\cS$ has independent edges
\begin{equation}
\bA[\cS]_{ij} \sim \Ber(p(X_i^\top X_j + 1)) \label{eq:intro-planted-part}
\end{equation}
where $X_i$'s are i.i.d.\ uniformly random vectors on the unit sphere $\bbS^{d-1}$ in $\R^d$ for some $d \geq 3$.
It is easily seen that the planted subgraph has edge density $p$, same as $G(n,p)$, but triangle density $p^3(1+1/d^2)$, higher than $p^3$ in the ambient \erdosrenyi{} graph (see Lemma~\ref{lem:triangle-denser}). 
Our goal is to recover the planted subgraph, i.e., to estimate $\cS$, by leveraging the elevated triangle density.

Towards this end, we consider the matrix 
\begin{equation}
\bM := \bar \bA^2 \circ \bar \bA ,
\label{eq:def-triangle-matrix-M}
\end{equation}
where $\circ$ denotes the Hadamard product and $\bar \bA$ denotes the centered adjacency matrix defined by $\bar \bA_{ij} = \bA_{ij} - p$ for $i \ne j$ and $\bar \bA_{ii} = 0$. 
The reason for considering this matrix is that each entry $\bM_{ij} = \sum_{\ell=1}^n \bar \bA_{ij} \bar \bA_{i\ell} \bar \bA_{\ell j}$ is the \emph{local, signed triangle count} at $(i,j)$. As a result, $\bM$ has an elevated mean over entries in the planted part $\cS$ (see Lemma~\ref{lem:M-expectation} for a precise statement). Given this crucial property, we can then recover $\cS$ from $\bM$ by running either (i) a spectral method that computes the leading eigenvector of $\bM$ and takes the indices of its $k$ largest entries, or (ii) a semidefinite program with $\bM$ as the input. 
Our main results show that the above two efficient methods achieve almost exact recovery and exact recovery of $\cS$, respectively, under the primary condition $k p^{3/4} \ge C n^{1/2}d$ for a sufficiently large constant $C>0$. 

When $p$ and $d$ are both constants, the above condition reduces to $k \gtrsim \sqrt{n}$, which is the same as the celebrated $\sqrt{n}$ computational threshold for the planted clique problem. This motivates us to ask whether an analogous statistical-to-computational gap exists for the planted triangle-dense subgraph model. We provide evidence supporting this gap: if $k = o(\sqrt{n})$, then no low-degree polynomial can estimate $\cS$ in the sense of \cite{schramm2022computational}; however, if $k$ is at least logarithmic in $n$, then an exponential-time algorithm finds the planted triangle-dense subgraph. 

Finally, the main technical challenge in proving theoretical guarantees for our algorithms lies in the spectral analysis of the matrix $\bM$ defined in \eqref{eq:def-triangle-matrix-M}. Such matrices are known as \emph{graph matrices} \citep{medarametla2016bounds} or \emph{motif adjacency matrices} \citep{benson2016higher}, which are special cases of \emph{polynomial random matrices} \citep{rajendran2023concentration} or \emph{matrix chaoses} \citep{bandeira2025matrix}. Standard matrix concentration inequalities do not apply directly to such graph matrices because their entries are polynomials of the entries of the adjacency matrix. In addition, the planted subgraph in our case is defined through latent points on a sphere as in \eqref{eq:intro-planted-part}, which further complicates the analysis. To tackle these challenges, our proofs combine decoupling inequalities \citep{de2012decoupling} with iterated matrix concentration inequalities (similar to \cite{bandeira2025matrix}) for the noise part, and use tools from spherical harmonics \citep{dai2013approximation} for the signal part. 


\subsection{Related work}\

\textbf{Planted dense subgraphs:}
There is an extensive literature on community detection \citep{fortunato2010community,arias2014community,abbe2015exact,mossel2015reconstruction} and the dense subgraph problem \citep{feige1997densest,bhaskara2010detecting,ames2015guaranteed}. 
Particularly, the computational thresholds for both detection and recovery have received considerable attention in recent years \citep{hajek2015computational,bresler2023detection,schramm2022computational,dhawan2025detection}. 
Moving beyond the assumption of inhomogeneous edge density in the observed graph, one may ask how to recover a planted subgraph that has the same expected edge density as the ambient graph, but a different expected density of another template graph. We initiate this study by considering the case of elevated triangle density. Both the spectral method and the semidefinite program we use are analogous to their counterparts for planted dense subgraph recovery \citep{hajek2016achieving}, with the centered adjacency matrix replaced by \eqref{eq:def-triangle-matrix-M}.

\textbf{Triangle-densest $k$-subgraph problem:} 
The triangle-densest $k$-subgraph problem refers to finding a subset of $k$ vertices in a given graph that induces the maximum number of triangles. 
\cite{konar2022triangle} proved that this problem is NP-hard and cannot be approximated efficiently in the worst case. 
This is also related to a few other optimization problems concerning triangle-densest subgraphs in the literature \citep{tsourakakis2015k,samusevich2016local,wang2020finding}. 
It is worth noting that the semidefinite program we use for the triangle-densest $k$-subgraph problem is quite different from the convex relaxation proposed by \cite{konar2022triangle} and is more in line with the programs in \cite{goemans1995improved,abbe2015exact,hajek2016achieving} that are often used for planted models.

\textbf{Random graphs with elevated triangle density:}
Triangle density is a central concept in social network analysis. For example, the well-known \emph{clustering coefficient} is defined to be the triangle density in the neighborhood of a vertex \citep{watts1998collective}, and the related concept of the number of edges in an \emph{egonet} also concerns the triangle count in a neighborhood \citep{bhadra2018detecting}. 
To model a subgraph with an elevated triangle density, we use a random geometric graph, which belongs to the general class of latent space models that has long been studied in network analysis \citep{hoff2002latent}. 
The node feature vectors in a latent space give rise to a clustering effect, and thus to a higher triangle density.

The RGG model we consider is particularly motivated by a recent line of work on the detection of latent geometry in random graphs using the signed triangle count \citep{bubeck2016testing,liu2022testing,bangachev2025random,mao2026random}. The linear kernel we assume is a special case of general smooth kernels considered by \citep{liu2023probabilistic}. 
As noted above, the excess triangle density of the planted subgraph compared to the ambient \erdosrenyi{} graph is $p^3/d^2$, so the dimension $d$ is a modeling parameter that determines the triangle density. 

There are other models for random graphs with more triangles. One direct way is to plant extra triangles in an \erdosrenyi{} graph \citep{bresler2023algorithmic}. It is also possible to consider \erdosrenyi{} graphs conditional on having a higher triangle density \citep{chatterjee2011large}. 
Moreover, a random graph with desired edge and triangle densities can be generated from the exponential random graph model \citep{chatterjee-diaconis} or from a suitable choice of a graphon \citep{kenyon2017phases}. 
It is interesting and challenging to develop a general theory for triangle-dense subgraphs that applies to a range of models.

\textbf{Graph matrices and matrix chaos:}
An important part of our work is a spectral analysis of the graph matrix $\bM$ defined in \eqref{eq:def-triangle-matrix-M} and is therefore related to the literature on graph matrices and matrix chaos. 
Also known as motif adjacency matrices, they have been used in higher-order spectral methods in network analysis \citep{benson2016higher,paul2023higher}. Our work extends this line of research. 
One step in our proofs, \Cref{thm:MainNoiseTerm}, is bounding the spectral norm of $\bZ^2 \circ \bZ$ where $\bZ$ is a Wigner matrix. 
The spectral norm of such a graph matrix has been studied by \cite{medarametla2016bounds,ahn2016graph,rajendran2023concentration}, but their bounds are not sufficiently sharp to yield the desired condition without extra logarithmic factors. 
More recently, a remarkable work by \cite{bandeira2025matrix} provides a general way to obtain sharp bounds on the expected norms of matrix chaos. Our proof follows a similar recipe, which consists of matrix decoupling and iterated concentration inequalities, and establishes a high-probability bound.


\section{Main results}


\subsection{Planted triangle-dense subgraph}
For a positive integer $n$ and an edge density parameter $p = p_n \in (0,1)$, let $G(n,p)$ denote the \erdosrenyi{} random graph model. For a graph $\bA$ on $n$ vertices, we identify $\bA$ with its adjacency matrix in $\R^{n \times n}$. 
We assume the graph is simple, so that $\bA_{ii} = 0$ for $i \in [n]$. 
For a subset of vertices $\cS \subset [n]$, we let $\bA[\cS]$ denote the subgraph of $\bA$ induced by $\cS$. Let us first define a model for planting a fixed subgraph in an \erdosrenyi{} graph.

\begin{definition}[Planting a subgraph in an \erdosrenyi{} graph]
For $k = k_n \in [n]$, fix a graph $\bH$ on $k$ vertices. Let $\cS$ be a uniformly random subset of $[n]$ of size $|\cS| = k$. Given $\bA_0 \sim G(n,p)$, we replace $\bA_0[\cS]$ by $\bH$ by assigning the vertices of $\bH$ uniformly randomly to the vertices in $\cS$. We denote the resulting graph by $\bA$ and write $\bA \sim G(n,p; \bH)$.
\end{definition}

In this work, we focus on a subgraph $\bH$ whose expected edge density is also $p$, matching that of the ambient \erdosrenyi{} graph, and whose triangle density is higher than the \erdosrenyi{} graph.
To this end, we model $\bH$ as a random geometric graph generated from a linear kernel of inner products, or simply a random inner product graph, defined as follows.

\begin{definition}[Random inner product graph]\label{dfn:RIPG}
For positive integers $k, d \ge 3$, let $X_1, \dots, X_k$ be i.i.d.\ random vectors from the uniform distribution on $\bbS^{d-1}$, the unit sphere in $\R^d$. For $p \in (0, 1/2]$, suppose that the random graph $\bH$ on $k$ vertices has independent edges $\bH_{ij} \sim \Ber(p(X_i^\top X_j + 1))$ up to symmetry $\bH_{ij} = \bH_{ji}$. We also zero the diagonal, $\bH_{ii} = 0$. We write $\bH \sim \GIP(k,p,d)$. 
\end{definition}

\noindent 
One can potentially consider a different kernel $\kappa(X_i^\top X_j)$ that takes values in $[0,1]$ in the above definition of the random inner product graph (see, e.g., \cite{liu2023probabilistic}), but we focus on the linear kernel $\kappa(t) = p(t+1)$ for simplicity. 

With the above definitions, our observation is a graph $\bA \sim G(n, p; \bH)$, where here $\bH \sim \GIP(k, p, d)$. In other words, we observe a random graph $\bA$ on $n$ vertices which contains a subgraph $\bH$ on $k$ vertices planted in an unknown location $\cS$.
Note that the planted subgraph $\bH$ has edge density $p$ (averaged over the randomness in the $X_i$'s). To see that it has a higher triangle density, let us note the following.

\begin{lemma} \label{lem:triangle-denser}
For distinct $i,j,\ell \in [n]$, we have $\E[\bH_{ij} \bH_{j\ell} \bH_{\ell i}] = p^3 (1+1/d^2)$. 
\end{lemma}

\begin{proof}
Since $X_i,X_j,X_\ell$ are i.i.d. uniform over $\bbS^{d-1}$, we have $\E[X_iX_i^\top] = \bI_d/d$ and $\E[X_i^\top X_j] = 0$, and similarly for other indices. Thus all the cross-terms vanish in the following expansion: 
\begin{align*}
\E[\bH_{ij} \bH_{j\ell} \bH_{\ell i}] &= p^3 \E[(X_i^\top X_j + 1) (X_j^\top X_\ell + 1) (X_\ell^\top X_i + 1)] \\
&= p^3 \E[X_i^\top X_j X_j^\top X_\ell X_\ell^\top X_i] + p^3 \\
&= p^3 \E\Big[ X_i^\top \Big( \frac 1d \bI_d \Big) \Big( \frac 1d \bI_d \Big) X_i \Big] + p^3 \\
&= p^3 / d^2 + p^3 \,.
\end{align*}
\end{proof}

\noindent 
Since $\bH_{ij} \bH_{j\ell} \bH_{\ell i}$ indicates the presence of the triangle with vertices $i,j,\ell$ in $\bH$, we see that $\bH$ indeed has a higher triangle density than the rest of $\bA$ (which has triangle density $p^3$).


\subsection{Spectral method}\label{subsec:spectralmethod}
Our spectral method for recovering the planted subgraph, i.e., estimating $\cS$, is as follows. 
Let $\bM$ be the matrix defined by \eqref{eq:def-triangle-matrix-M}. 
Let $\hat u = \hat u(\bM)$ denote the leading eigenvector of $\bM$, i.e., the eigenvector associated with the largest eigenvalue $\hat \lambda$ of $\bM$. By convention $\hat u$ is a unit vector chosen with arbitrary sign; compare to \eqref{eq:top-eigenvector-estimation-error} below. 
We estimate the planted location $\cS$ by the set $\hat \cS$ of indices $i \in [n]$ such that $|\hat u_i|$ is among the largest $k$ entries of $\hat u$ in absolute value (with ties broken arbitrarily).

Before presenting our theoretical guarantees for $\hat u$ and $\hat \cS$, let us first explain the intuition behind the spectral method. To see that the matrix $\bM$ captures the the triangle density of the graph, it is helpful to consider the non-centered version $\bA^2 \circ \bA$ whose entry
$$
(\bA^2 \circ \bA)_{ij} = \sum_{\ell = 1}^n \bA_{i \ell} \bA_{\ell j} \bA_{ij} 
$$
is precisely the number of triangles in $\bA$ that contain the edge $(i,j)$. Therefore, the entries of $\bA^2 \circ \bA$ are local triangle counts, while the entries of $\bM = \bar \bA^2 \circ \bar \bA$ are local signed triangle counts, where the centering has proved to be effective for variance reduction in various settings \citep{bubeck2016testing}. 
More formally, it is not hard to see the following.

\begin{lemma} \label{lem:M-expectation}
Let $\bM$ be as defined above. Then we have
$$
\E[\bM] = \frac{p^3 (k-2)}{d^2} (\bsone_\cS \bsone_\cS^\top - \bI_\cS) ,
$$
where $\bsone_\cS \in \R^n$ is the indicator vector of the set $\cS \subset [n]$, and $\bI_\cS \in \R^{n \times n}$ denotes the matrix with a $k \times k$ identity principal minor indexed by $\cS$ and zeros elsewhere.
\end{lemma}

\begin{proof}
The diagonal of $\bM$ is zero by definition. Consider $i \ne j$. If $i$ or $j$ is not in $\cS$, then it is easily seen that $\E[\bar \bA_{ij}] = 0$, $\E[\bar \bA_{i \ell} \bar \bA_{\ell j} \bar \bA_{ij}] = 0$ for all $\ell \in [n]$, and so $\E[\bM_{ij}] = \E[(\bar \bA^2 \circ \bar \bA)_{ij}] = 0$. If $i, j \in \cS$, then
$$
\E[\bM_{ij}] = \sum_{\ell = 1}^n \E[ \bar \bA_{i \ell} \bar \bA_{\ell j} \bar \bA_{ij} ] 
= \sum_{\ell \in \cS \setminus \{i,j\}} p^3 \E[ X_i^\top X_j X_j^\top X_\ell X_\ell^\top X_i ] 
= (k-2)p^3/d^2 
$$
where the last step has been shown in the proof of Lemma~\ref{lem:triangle-denser}.
\end{proof}

Thus the spectrum of $\E[\bM]$ is explicit: it consists of an eigenvalue of $\frac{p^3(k-2)}{d^2}(|\cS|-1) = \frac{p^3(k-2)(k-1)}{d^2}$ with corresponding eigenvector $\bsone_{\cS}$, an eigenvalue of $-\frac{p^3(k-2)}{d^2}$ with multiplicity $|\cS|-1 = k-1$, and a kernel of dimension $n - |\cS| = n-k$. If $k \to \infty$, then the first eigenvalue is the leading one, meaning that $\E[\bM]$ has leading eigenvector $\bsone_{\cS}$. In other words, one can recover $\cS$ from the top eigenvector of $\E[\bM]$. If $\bM$ is close to its expectation, then one might hope that the top eigenvector of $\bM$ would give approximate recovery of $\cS$. 
Our first main theorem, below, shows that this heuristic is correct, even though its proof (deferred to \Cref{sec:proofs}) shows that $\bM$ is \emph{not} just ``$\E[\bM]$ plus small noise.''

\begin{theorem} \label{thm:main-spectral}
Let $\bA \sim G(n, p; \bH)$ be given, where $\bH \sim \GIP(k, p, d)$. Recall that $\hat u$ is the leading eigenvector of $\bM$ defined by \eqref{eq:def-triangle-matrix-M}. 
For any $\epsilon > 0$, there is $C > 0$ depending only on $\eps$ with the following property. 
If 
\begin{equation}
k p^{3/4} \ge C n^{1/2}d 
\label{eq:main-condition-spectral}
\end{equation}
and $n p^{3/2} \ge (\log n)^{3/2}$,
then the following holds with probability at least $1-n^{-10}$: We have
\begin{equation}
\min_{\zeta \in \{\pm 1\}} \Big\| \zeta \hat u - \frac{1}{\sqrt{k}} \bsone_{\cS} \Big\| \le \epsilon ,
\label{eq:top-eigenvector-estimation-error}
\end{equation}
and for any set\footnote{There can be more than one such set, because of possible ties when $\hat{u}$ has multiple entries with the same magnitude.} $\hat \cS$ consisting of exactly $k$ indices $i \in [n]$ such that $|\hat u_i|$ is among the largest $k$ entries of $\hat u$ in absolute value, we have
\begin{equation}
|\hat \cS \triangle \cS| \le 8 \eps^2 k .
\label{eq:planted-set-estimation-error}
\end{equation}
\end{theorem}

We note that the secondary assumption $n p^{3/2} \ge (\log n)^{3/2}$ is mild. Moreover, it is, up to a logarithmic factor, subsumed by the main condition $k p^{3/4} \ge C n^{1/2}d$ in \eqref{eq:main-condition-spectral}. Indeed, if we have $n \ge k \ge Cn^{1/2}d/p^{3/4}$, then $np^{3/2} \geq Cd^2 \geq C$.

To provide heuristics\footnote{This simple heuristic is analogous to the case of planted clique recovery: the celebrated $k=\sqrt{n}$ threshold can be seen by comparing the number of additional neighbors of a planted vertex (order $k$) to the standard deviation of the total number of neighbors of that vertex in $G(n,1/2)$ (order $\sqrt{n}$). In our case, we consider (signed) triangles instead of neighbors (i.e., incident edges).} for the main condition \eqref{eq:main-condition-spectral}, suppose that vertex $1$ is planted, i.e., $1 \in \cS$, and consider the statistic
$$
T := \sum_{j,\ell=1}^n \bar \bA_{1j} \bar \bA_{1\ell} \bar \bA_{\ell j} = \sum_{j=1}^n \bM_{1j} ,
$$
which is twice the signed count of triangles containing vertex $1$. 
By \Cref{lem:M-expectation}, we have $\E[T] = \frac{p^3 (k-2)}{d^2} (k-1) \asymp \frac{p^3 k^2}{d^2}$. On the other hand, the expectation of $T$ under $G(n,p)$ is zero, and the variance of $T$ under $G(n,p)$ is 
\begin{align*}
\Var_{G(n,p)}(T) 
&= \sum_{j,\ell=1}^n \sum_{j',\ell'=1}^n \Cov_{G(n,p)}(\bar \bA_{1j} \bar \bA_{1\ell} \bar \bA_{\ell j}, \bar \bA_{1j'} \bar \bA_{1\ell'} \bar \bA_{\ell' j'}) \\
&= 2 \sum_{j,\ell=1}^n \E_{G(n,p)}[\bar \bA_{1j}^2 \bar \bA_{1\ell}^2 \bar \bA_{\ell j}^2]
\asymp n^2 p^3 .
\end{align*}
Comparing $\frac{p^3 k^2}{d^2}$ to $\sqrt{n^2 p^3}$ yields the threshold $k p^{3/4} = n^{1/2} d$. 

This heuristic can potentially be turned into a rigorous analysis of the algorithm that computes the local signed triangle count at each vertex (which would then be similar to the egonet method in \cite{bhadra2018detecting}, except that the triangle count there is not signed). However, this method is unlikely to achieve the tight condition \eqref{eq:main-condition-spectral}, because the analysis of such a local algorithm typically entails a union bound which incurs an additional logarithmic factor. On the other hand, the spectral method aggregates information globally and is thus able to achieve the threshold \eqref{eq:main-condition-spectral}.

The above spectral method, theoretical guarantees, and heuristics are all analogous to those for the planted clique problem. However, the analysis of the spectral method is significantly more challenging in this case, because we need to analyze the matrix $\bM$ whose entries are polynomials of random variables. 


\subsection{Semidefinite programming}
Theorem~\ref{thm:main-spectral} shows that the spectral method achieves $(1-\eps)$-recovery of the planted part $\cS$ for any small $\epsilon > 0$ if $k p^{3/4} \ge C n^{1/2}d$ for a large enough $C$. 
One further direction is to design a rounding step for the spectral estimator to obtain exact recovery of $\cS$, or to develop an iterative procedure that uses the spectral method as a module to achieve recovery under an even weaker condition, $k p^{3/4} \ge c n^{1/2} d$ for any small $c > 0$, analogous to \cite{alon1998finding}.
Instead of pursuing these directions, we focus in this work on global, one-shot algorithms and show that a semidefinite program based on the matrix $\bM$ defined in \eqref{eq:def-triangle-matrix-M} achieves exact recovery, assuming the same main condition as the spectral method.

Let $\bJ$ be the $n \times n$ matrix of all ones. 
Consider the following semidefinite program:
\begin{equation}\label{eqn:SDPrelaxation}
\arraycolsep=1.4mm
\begin{array}{ll}
\max & \inner{\bM,\bX} \\[5pt]
{\textrm{s.t.}} & \bX\succeq0 \,, \\[5pt]
& \bX_{ij}\geq0 \,, \qquad i,j=1,\dots,n \,, \\[5pt]
& \bX_{ii}\leq1 \,, \qquad i=1,\dots,n \,, \\[5pt]
& \Tr(\bX)=k \,, \\[5pt]
& \inner{\bJ,\bX}=k^2 \,.
\end{array}
\end{equation}
This program, with $\bA$ in place of $\bM$, was proposed by \cite{hajek2016achieving} to solve the planted dense subgraph problem. 
The following theorem shows that the same program with $\bM$ as the data matrix succeeds at exact recovery of a planted triangle-dense subgraph.

\begin{theorem}\label{thm:SDPmain}
Let $\bA \sim G(n, p; \bH)$ be given, where $\bH \sim \GIP(k, p, d)$. Let $\bM$ be defined by \eqref{eq:def-triangle-matrix-M}. 
If \eqref{eq:main-condition-spectral} holds for a sufficiently large constant $C>0$ and $np^{3/2}=\omega((d\log n)^2)$, then $\bsone_{\cS} \bsone_{\cS}^\top$ is the unique optimizer of \eqref{eqn:SDPrelaxation} with probability at least $1-n^{-10}$.
\end{theorem}

Similar to the discussion after Theorem~\ref{thm:main-spectral}, the condition $np^{3/2}=\omega((d\log n)^2)$ is a mild assumption which does not involve the subgraph size $k$, and it is subsumed by the main condition $k p^{3/4} \ge C n^{1/2}d$ in \eqref{eq:main-condition-spectral} up to a logarithmic factor.


\subsection{Statistical-to-computational gap}

The main condition \eqref{eq:main-condition-spectral} reduces to $k \gtrsim \sqrt{n}$ if $p$ and $d$ are both constants. Given the analogy to the $k = \sqrt{n}$ computational threshold for the planted clique problem, we provide evidence demonstrating that this is also the computational threshold for the recovery of a planted triangle-dense subgraph, and that there is a statistical-to-computational gap in our case, too.

For the computational hardness result, we follow \cite{schramm2022computational} to show a recovery lower bound for low-degree polynomial algorithms. 
Let $\theta:=\bsone\{1\in\cS\}$ be the indicator for the event that the first vertex belongs to the planted subgraph. 
We will show that if $k=o(\sqrt{n})$, then one cannot estimate $\theta$ much better than with a trivial estimator in the following sense.  
Let $\bbR[\bA]_{\leq D}$ be the set of real polynomials of degree at most $D$ whose variables are the entries of $\bA$. 
The \emph{degree-$D$ maximum correlation} is defined as 
$$\CorrleD:=\sup_{f\in\bbR[\bA]_{\leq D}} \frac{\bbE[f(\bA)\cdot\theta]}{\sqrt{\bbE[f(\bA)^2]}}\,.$$

\begin{theorem}\label{thm:CompLwrBd}
Let $\bA\sim G(n,p;\,\bH)$ be given, where $\bH\sim \GIP(\tilde{k},p,d)$ and $\tilde{k}$ is a binomial random variable with parameters $n$ and $r:=k/n$. 
Let $p\in(0,1/2)$, $d\geq2$, and $D=o\big(\big(\frac{\log n}{\log\log n}\big)^2\big)$.
If $k\leq n^{1/2-\epsilon}$ for a constant $\epsilon>0$ then
$$\CorrleD = (1+o(1))r \,.$$ 
\end{theorem}

\noindent \Cref{thm:CompLwrBd} asserts that no degree-$D$ estimator achieves asymptotically larger correlation than the trivial estimator $f\equiv r$, since $\frac{\bbE[r\cdot\theta]}{\sqrt{\bbE[r^2]}} = r$. 

To provide evidence for a statistical-to-computational gap, it remains to prove an information-theoretic upper bound. The result below shows that a logarithmically sized triangle-dense subgraph can be partially recovered with high probability by a computationally inefficient estimator. More precisely, if $k=\Theta(n)$ then let $\hat{\cS}$ be a uniformly random $k$-subset of $[n]$, while if $k=o(n)$ let
\begin{equation}\label{eqn:Shat-kon-dfn}
\hat{\cS} \in \argmax_{T\subseteq[n],\,|T|=k}\sum_{i,j,l\in T}\bar{\bA}_{ij}\bar{\bA}_{jl}\bar{\bA}_{li}\,,
\end{equation}
which maximizes the signed triangle count over all $k$-subgraphs of $\bA$.

\begin{theorem}\label{thm:ITUpperBound}
Let $\bA \sim G(n, p; \bH)$ be given, where $\bH \sim \GIP(k, p, d)$. 
For all constant $p\in(0,1/2]$ and $d\geq2$, there exist constants $C=C(p,d)>0$ and $\epsilon=\epsilon(d)>0$ such that if $k\geq C\log n$, then $|\cS\cap\hat{\cS}|\geq \epsilon k$ with probability $1-o(1)$.
\end{theorem}
\noindent
Recall that $r = k/n$ and $\theta=\bsone\{1\in\cS\}$. If we define $\hat \theta := \bsone\{1 \in \hat \cS\}$, then by symmetry,
$$
\frac{\bbE[\hat \theta\cdot\theta]}{\sqrt{\bbE[\hat \theta^2]}} = \frac{\bbP\{1 \in \cS \cap \hat \cS\}}{\sqrt{\bbP\{1 \in \hat \cS\}}} \geq \frac{(1-o(1)) \epsilon k/n}{\sqrt{r}} = (1-o(1)) \epsilon \sqrt{r} \,.
$$
This is much larger than $\CorrleD = (1+o(1)) r$ in \Cref{thm:CompLwrBd}, so there is a statistical-to-computational gap for recovery in the regime $C \log n \le k \le n^{1/2-\epsilon}$.

The proofs of \Cref{thm:CompLwrBd,thm:ITUpperBound} are given in \Cref{sec:computational-lower bound,sec:ITUpperBound} respectively.


\section{Analysis of the algorithms} \label{sec:proofs}


\subsection{Proof of \Cref{thm:main-spectral}}
We now prove \Cref{thm:main-spectral}, our main theoretical result for the spectral method. 
Since the spectral method is equivariant with respect to a relabeling of the vertices of the observed graph, we may assume that $\cS = [k]$ throughout the proof without loss of generality. We also introduce the following notation: For any matrix $\bB \in \R^{k \times k}$, we let $\bB^\# \in \R^{n \times n}$ be the matrix with $\bB$ as its top-left $k \times k$ principal minor and zeros elsewhere. In the remainder, we will frequently use that, for any matrices $\bA, \bB \in \R^{k \times k}$, we have $\bA^\# \bB^\# = (\bA \bB)^\#$ and $\bA^\# \circ \bB^\# = (\bA \circ \bB)^\#$.

Since the proof of Theorem~\ref{thm:main-spectral} involves a long list of auxiliary matrices and lemmas, we first provide a proof outline before delving into the details. At a high level, Theorem~\ref{thm:main-spectral} follows from a spectral perturbation analysis using, in particular, the classical Davis--Kahan theorem (originally \cite{davis1970rotation}, or see, e.g., Theorem 4.5.5 in \cite{vershynin2018high} for a textbook treatment).

\begin{lemma}[Davis--Kahan] \label{lem:davis-kahan}
For symmetric matrices $\bA_1, \bA_2 \in \R^{n \times n}$, let $\{(\lambda_i(\bA_1), u_i(\bA_1)) : i \in [n]\}$ be the set of eigenpairs of $\bA_1$ where the eigenvalues are ordered decreasingly, i.e., $\lambda_1(\bA_1) \ge \cdots \ge \lambda_n(\bA_1)$, and every $u_i(\bA_1)$ is a unit vector. Define the notation for $\bA_2$ in the same way. Then we have
$$
\min_{\zeta \in \{\pm 1\}} \|\zeta u_1(\bA_2) - u_1(\bA_1)\| \le \frac{2 \sqrt{2} \, \|\bA_2-\bA_1\|}{\lambda_1(\bA_1) - \lambda_2(\bA_1)} ,
$$
where the norm on the left-hand side is the usual Euclidean distance on $\R^n$ and the norm on the right-hand side is the operator norm with respect to this distance.
\end{lemma}

\noindent 
To apply the Davis--Kahan theorem, we will take $\bA_2 = \bM$. It would be nice if we could take $\bA_1$ to be $\E[\bM]$, because $\E[\bM]$ has top eigenvector $\bsone_{\cS}/\sqrt{k}$, and a large top eigengap (its top eigenvalue is $k$ times larger than the rest, and $k$ is growing; see the discussion after \Cref{lem:M-expectation}). However, it turns out that $\|\bM - \E[\bM]\|$ is at the same order as this eigengap; thus $\bA_1 = \E[\bM]$ is not a good choice in Davis--Kahan. Instead, we will define a \emph{random} matrix $\widetilde{\bM}$ below, whose top eigenvector is deterministically $\bsone_{\cS}/\sqrt{k}$. Unlike $\E[\bM]$, the largest and second-largest eigenvalues of $\widetilde{\bM}$ (which are deterministic) are of the same order as one another; however, we can get a better bound on $\|\bM - \widetilde{\bM}\|$, so that we will ultimately choose $\bA_1 = \widetilde{\bM}$ in Davis--Kahan.

Towards this end, the first step is to decompose the matrix $\bM$ into a signal term plus a few noise terms.
With $\bX := [X_1 \cdots X_k]^\top \in \R^{k \times d}$, we let $\bK \in \R^{k \times k}$ denote the matrix $\bX \bX^\top$ with its diagonal replaced by zeros, i.e., 
\begin{equation}
\bK := \bX \bX^\top - \bI_k. 
\label{eq:def-K}
\end{equation}
Note that $p \bK^\# = \E[\bar \bA \mid \bX]$. 
We further define 
$$
\bW := \bar \bA - p \bK^\#.
$$
Therefore, 
\begin{align}
\bM = \bar \bA^2 \circ \bar \bA 
&= p^3 (\bK^2 \circ \bK)^\# + p^2 (\bK^2)^\# \circ \bW + p \bW^2 \circ \bK^\# + \bW^2 \circ \bW \label{eq:M-noise-decomposition} \\
& \qquad + p^2 (\bK^\# \bW + \bW \bK^\#) \circ \bK^\# + p (\bK^\# \bW + \bW \bK^\#) \circ \bW . \notag
\end{align}
In the above decomposition, $p^3 (\bK^2 \circ \bK)^\#$ can be seen as the signal term, while all the other terms are noise. We will bound the spectral norms of all the noise terms in \Cref{sec:noise-terms}.

To understand the matrix $\bK^2 \circ \bK$, we will first show that it is close to $\frac{k}{d} (\bX \bX^\top) \circ (\bX \bX^\top)$, in Lemma \ref{lem:K2-K-K-K}. Next, we use the theory of spherical harmonics to write 
\begin{equation}
\frac{1}{k} (\bX \bX^\top) \circ (\bX \bX^\top) = \bY \bLam \bY^\top , \label{eq:K-K-spherical-harmonics}
\end{equation}
where the columns of $\bY$ are defined from spherical harmonics and $\bLam$ is a diagonal matrix consisting of the corresponding eigenvalues (ordered decreasingly), both to be specified in Lemma~\ref{lem:Y-Lambda-eq}. Finally, we orthonormalize the columns of $\bY$ to obtain a matrix $\bQ$, and show that $\bY \bLam \bY^\top$ is close to $\bQ \bLam \bQ^\top$, in Lemma \ref{lem:QR-decomposition}. 
We will see that by definition the first column of $\bQ$ is exactly $\frac{1}{\sqrt{k}} \bsone_k$ where $\bsone_k$ is the all-ones vector in $\R^k$. 

Putting it together, we eventually choose $\bA_1 = \widetilde{\bM}$ in Lemma~\ref{lem:davis-kahan}, where
\begin{equation}
\widetilde{\bM} := \frac{p^3 k^2}{d} (\bQ \bLam \bQ^\top)^\# ,
\label{eq:def-tilde-M}
\end{equation}
which then implies that
\begin{equation}
\min_{\zeta \in \{\pm 1\}} \Big\| \zeta u_1(\bM) - \frac{1}{\sqrt{k}} \bsone_{\cS} \Big\| 
\le \frac{2 \sqrt{2} \, \|\bM - \widetilde{\bM}\|}{\lambda_1(\widetilde{\bM}) - \lambda_2(\widetilde{\bM})} .
\label{eq:apply-davis-kahan}
\end{equation}
Moreover, by \eqref{eq:M-noise-decomposition} and \eqref{eq:K-K-spherical-harmonics}, we obtain
\begin{align}
\Big\| \bM - \widetilde{\bM} \Big\| 
&= \Big\| \Big( p^3 \bK^2 \circ \bK - \frac{p^3 k}{d} (\bX \bX^\top) \circ (\bX \bX^\top) + \frac{p^3 k^2}{d} \bY \bLam \bY^\top - \frac{p^3 k^2}{d} \bQ \bLam \bQ^\top \Big)^\# \notag \\
&\qquad + p^2 (\bK^2)^\# \circ \bW + p \bW^2 \circ \bK^\# + \bW^2 \circ \bW \notag \\
& \qquad + p^2 (\bK^\# \bW + \bW \bK^\#) \circ \bK^\# + p (\bK^\# \bW + \bW \bK^\#) \circ \bW \Big\| \notag \\
&\le p^3 \Big\| \bK^2 \circ \bK - \frac{k}{d} (\bX \bX^\top) \circ (\bX \bX^\top) \Big\| + \frac{p^3 k^2}{d} \| \bY \bLam \bY^\top - \bQ \bLam \bQ^\top \| \label{eq:all-terms-spectral-norm} \\
&\qquad + p^2 \| (\bK^2)^\# \circ \bW \| + p \| \bW^2 \circ \bK^\# \| + \| \bW^2 \circ \bW \| \notag \\
& \qquad + p^2 \| (\bK^\# \bW + \bW \bK^\#) \circ \bK^\# \| + p \| (\bK^\# \bW + \bW \bK^\#) \circ \bW \| . \notag
\end{align}
In \Cref{sec:additional-proofs-spectral}, we will bound each of the terms in \eqref{eq:all-terms-spectral-norm} to obtain the following result.

\begin{proposition} \label{prop:main-spectral}
The matrix $\widetilde{\bM}$ defined in \eqref{eq:def-tilde-M} has top two eigenvalues equal to $\frac{p^3 k^2}{d^2}$ and $\frac{p^3 k^2}{d^2(d/2+1)}$ respectively, and top eigenvector equal to $\frac{1}{\sqrt{k}} \bsone_{\cS}$. 
Moreover, for any $\epsilon > 0$, there is $C > 0$ depending only on $\eps$ such that the following holds. 
If $n p^{3/2} \ge (\log n)^{3/2}$ and $k p^{3/4} \ge C n^{1/2}d$, 
then with probability at least $1-n^{-10}$,  we have
$$
\frac{d^2}{p^3 k^2} \left\|\bM - \widetilde{\bM}\right\| \le \epsilon .
$$
\end{proposition}

Then \eqref{eq:top-eigenvector-estimation-error} is an immediate consequence of \eqref{eq:apply-davis-kahan} together with \Cref{prop:main-spectral}. 
To prove \eqref{eq:planted-set-estimation-error}, let us focus on the case $\Big\| \hat u - \frac{1}{\sqrt{k}} \bsone_{\cS} \Big\| \le \eps$, and the case $\zeta = -1$ is analogous. 
We have
$$
\Big\| \hat u - \frac{1}{\sqrt{k}} \bsone_{\cS} \Big\|^2 = \sum_{i \in \cS} \Big( \hat u_i - \frac{1}{\sqrt{k}} \Big)^2 + \sum_{i \in [n] \setminus \cS} \hat u_i^2 \le \eps^2.
$$
If $|\hat u_i| \le \frac{1}{2 \sqrt{k}}$ for all $i \notin \hat{\cS}$, then 
$$
\eps^2 \ge \sum_{i \in \cS \setminus \hat{\cS}} \Big( \hat u_i - \frac{1}{\sqrt{k}} \Big)^2 \ge \frac{1}{4k} |\cS \setminus \hat{\cS}| = \frac{1}{8k} |\cS \triangle \hat{\cS}| . 
$$
If $|\hat u_i| > \frac{1}{2 \sqrt{k}}$ for some $i \notin \hat{\cS}$, then we must have $|\hat u_i| > \frac{1}{2 \sqrt{k}}$ for all $i \in \hat{\cS}$. Therefore,
$$
\eps^2 \ge \sum_{i \in \hat{\cS} \setminus \cS} \hat u_i^2 \ge \frac{1}{4k} |\hat{\cS} \setminus \cS| = \frac{1}{8k} |\cS \triangle \hat{\cS}| .
$$
In either case, we have $|\cS \triangle \hat{\cS}| \le 8 k \eps^2$. 


\subsection{Proof of \Cref{thm:SDPmain}}

We again assume that $\cS = [k]$ without loss of generality, and let $\xi := \bsone_{[k]}$ for brevity.
Let us start with a set of deterministic conditions that guarantee optimality.

\begin{lemma}[Lemma~19 in \cite{pmlr-v49-hajek16}]
\label{lemma:Dual}
Let $\bX^\ast=\xi\xi^\top$. If there exist $n\times n$ matrices $\bD=\diag(\bd)$ where $\bd$ is an entrywise nonnegative vector, $\bS\succeq0$, and $\bB\geq0$, and $\eta,\lambda \in \R$ such that
\begin{align*}
& \bS + \bM + \bB - \eta \bI_n - \lambda \bJ - \bD = 0 \,, \\[3pt]
& \inner{\bS,\bX^\ast} = 0 \,, \\[3pt]
& \inner{\bB,\bX^\ast} = 0 \,, \\[3pt]
& \inner{\bD,\bX^\ast-\bI_n}=0 \,, \\[3pt]
& \lambda_{n-1}(\bS) > 0 \,, 
\end{align*}
then $\bX^\ast$ is the unique maximizer in \eqref{eqn:SDPrelaxation}.
\end{lemma}

We now prove \Cref{thm:SDPmain}. 
Let $\delta:=n^{-10}$. It suffices to define matrices $\bD$, $\bS$, and $\bB$, and real numbers $\eta,\lambda>0$ that satisfy the conditions of \Cref{lemma:Dual} with probability at least $1-\delta$. We claim that the following assignments suffice:
\begin{equation}
\label{eqn:semidefinite-definitions}
\arraycolsep=0.4mm
\begin{array}{llll}
\lambda &:= \frac{1}{k}\max_{k<i\leq n} |(\bM\xi)_i| \,, \\[7pt]
\eta &:= \min_{1\leq i\leq k}(\bM\xi)_i-\lambda k \,, \\[7pt]
d_i &:= \bsone\{1\leq i\leq k\}\cdot\big((\bM\xi)_i -\lambda k-\eta\big) \,, & \hspace{3mm}\bD&:=\diag(d_1,\dots,d_n) \,, \\[7pt]
b_i &:= \bsone\{i>k\}\cdot\left(\lambda-\frac{1}{k}(\bM\xi)_i\right) \,, & \hspace{3mm} \bb&:=(b_1,\dots,b_n) \,, \\[7pt]
\bB &:= \bb\xi^\top+\xi\bb^\top \,, \\[7pt]
\bS &:= \bD-\bB-\bM+\eta \bI_n+\lambda\bJ \,.
\end{array}
\end{equation}
First notice that, by definition, we have that $d_i \geq 0$ for each $i$, and that $b_i \geq 0$ for each $i$, so that $\bB$ is indeed entrywise nonnegative. The condition $\inner{\bD,\bX^\ast-\bI_n}=0$ is immediate since $\bd$ is supported on its first $k$ entries while $\bX^\ast_{ii}=1$ for all $i\in[k]$. Since $\bb$ is orthogonal to $\xi$, we also have that $\ip{\bB, \bX^\ast} = \ip{\xi, \bB \xi} = 0$, as well as $\bS \xi = 0$; indeed, we have
\begin{align*}
&(\bS \xi)_i = d_i - (\bM \xi)_i + \eta + \lambda k = 0, \quad i \le k , \\
&(\bS \xi)_i = - k\lambda + (\bM \xi)_i - (\bM \xi)_i + k \lambda = 0, \quad i > k.
\end{align*}
In particular, $\ip{\bS, \bX^\ast} = \ip{\xi, \bS \xi} = 0$, and we conclude that, in order to show $\bS \succeq 0$ with $\lambda_{n-1}(\bS) > 0$, it suffices to show that $\ip{v,\bS v} > 0$ uniformly over unit vectors $v$ which are orthogonal to $\xi$. That is, it remains to prove that, with probability at least $1-\delta$,
$$\inf_{v \perp \xi, \|v\| = 1} \ip{v,\bS v} > 0.$$
For $\bS$, we note that, if $v$ is some unit vector which is orthogonal to $\xi$, we have
\[
    \ip{v,\bB v} = \ip{v,\bb \xi^T v} + \ip{v,\xi \bb^T v} = 2\ip{\xi,v}\ip{\bb, v} = 0,
\]
so that
\[
    \ip{v,\bS v} = \ip{v,\bD v} - \ip{v,\bB v} - \ip{v,\bM v} + \eta + \lambda \ip{v,1}^2 \geq \eta - \ip{v, \bM v} ,
\]
where we used that $\bD$ is positive semidefinite and that $\lambda \ge 0$. Thus it suffices to show that
\[
    \sup_{v \perp \xi, \|v\| = 1} \ip{v, \bM v} < \eta
\]
with probability at least $1-\delta$. For $\widetilde{\bM}$ defined in \eqref{eq:def-tilde-M}, \Cref{prop:main-spectral} shows that the top eigenvector of $\widetilde{\bM}$ is $\xi$, that the second-largest eigenvalue of $\widetilde{\bM}$ is $\frac{2p^3k^2}{d^2(d+2)}$, and that $\|\bM - \widetilde{\bM}\| \le \frac{1}{6}\frac{p^3 k^2}{d^2}$ with probability at least $1-\delta/2$. (\Cref{prop:main-spectral} shows this event occurs with probability at least $1-\delta$, but the proof is easily modified to yield $1-\delta/2$.) 
Therefore, 
\begin{equation}\label{eqn:sdp-sup-vMv}
    \sup_{v \perp \xi, \|v\| = 1} \ip{v, \bM v} \leq \sup_{v \perp \xi, \|v\| = 1} \ip{v, \widetilde{\bM} v} + \|\bM - \widetilde{\bM}\| 
    \le \frac{2p^3k^2}{d^2(d+2)} + \frac{1}{6} \frac{p^3 k^2}{d^2} \leq \frac{5}{6}\frac{p^3k^2}{d^2} \,.
\end{equation}
Applying \Cref{lemma:sdp-eta-lb} with $\epsilon=1/12$ yields that $\eta\geq\frac{11}{12}\frac{p^3k^2}{d^2}$ with probability at least $1-\delta/2$. Combining this with \eqref{eqn:sdp-sup-vMv} completes the proof.


\section{Future directions}

Our work leaves more interesting problems open than it solves. Recall that we have provided evidence showing that the condition $k \ge C \sqrt{n}$ is optimal among low-degree polynomial algorithms for constant $p$ and $d$, and provided a heuristic suggesting that the condition $k p^{3/4} \ge C n^{1/2}d$ is expected for methods based on triangle counts. However, the full picture of statistical and computational thresholds for the recovery problem remains an interesting open question. Moreover, we have not studied the detection problem for planted triangle-dense subgraphs, which may have thresholds different from those for the recovery problem.

In addition, to model a subgraph with the same edge density but higher triangle density than that in the ambient \erdosrenyi{} graph, the random geometric graph we use is just one natural model. As discussed in the introduction, other possibilities include a more general latent space model, an exponential family random graph model, or an \erdosrenyi{} model conditional on having more triangles, to name a few. It is an interesting direction to study a potentially more general model for planted triangle-dense subgraphs, and we think our spectral method can be competitive in other settings, too, thanks to its simplicity. 

Taking this one step further, one may study planted subgraphs characterized by a general homomorphism density. Is there a class of planted subgraph models where the subgraph has a higher $H$-density for a template graph $H$? How can we build a suitable graph matrix and analyze the corresponding spectral method? 
These questions pose substantial theoretical and algorithmic challenges.

\acks{S.v.d.P is supported by an Algorithms and Randomness Center Fellowship and an NSF Graduate Research Fellowship. C.M. is  supported in part by NSF CAREER Award 2338062. B.M. is supported in part by NSF grant DMS-1760471.}

\bibliography{ref}

\appendix
\markboth{SAM VAN DER POEL, CHENG MAO, AND BENJAMIN MCKENNA}{SPECTRAL RECOVERY OF A PLANTED TRIANGLE-DENSE SUBGRAPH}

\crefalias{section}{appendix} 


\section{Additional proofs for the spectral method} \label{sec:additional-proofs-spectral}

Continuing from \Cref{sec:proofs}, we bound each of the terms in \eqref{eq:all-terms-spectral-norm} to prove \Cref{prop:main-spectral}.


\subsection{Analyzing the signal term} \label{sec:signal-term-spherical-harmonics}

To analyze the matrix $\bK^2 \circ \bK$, we first show that it is close to $(\bX \bX^\top) \circ (\bX \bX^\top)$ up to normalization.

\begin{lemma} \label{lem:K2-K-K-K}
There is an absolute constant $C>0$ such that the following holds for any $\delta \in (0, 1)$. 
If $k \ge d + \log(1/\delta)$, then with probability at least $1 - \delta$, 
$$
\Big\| \bK^2 \circ \bK - \frac{k}{d} (\bX \bX^\top) \circ (\bX \bX^\top) \Big\| \le C  \frac{k^2}{d^2} \sqrt{\frac{d + \log(1/\delta)}{k}} .
$$
\end{lemma}

\begin{proof}
We have
\begin{align}
&\bK^2 \circ \bK - \frac{k}{d} (\bX \bX^\top) \circ (\bX \bX^\top) \notag \\
&= \Big( (\bX \bX^\top - \bI_k)^2 - \frac{k}{d} \bX \bX^\top \Big) \circ (\bX \bX^\top) - (\bX \bX^\top - \bI_k)^2 \circ \bI_k \notag \\
&= \Big( \bX \bX^\top \bX \bX^\top - \Big( \frac kd + 2 \Big) \bX \bX^\top \Big) \circ (\bX \bX^\top) + \bI_k \circ (\bX \bX^\top) - (\bX \bX^\top - \bI_k)^2 \circ \bI_k \notag \\
&= \Big( \bX \Big( \bX^\top \bX - \Big( \frac kd + 2 \Big) \bI_d \Big) \bX^\top \Big) \circ (\bX \bX^\top) - (\bX \bX^\top)^2 \circ \bI_k + 2 \bI_k , \label{eq:K-three-terms}
\end{align}
where the last equality holds because $\bX \bX^\top$ has all ones on its diagonal. 

Let us start with the first term in \eqref{eq:K-three-terms}. 
For any matrix $L \in \R^{k \times k}$, by Lemma~\ref{lem:hadamard-spectral-norm},
$$
\|L \circ (\bX \bX^\top)\| 
\le \max_{i \in [k]} X_i^\top X_i \cdot \|L\| 
= \|L\| .
$$
Therefore, 
$$
\Big\| \Big( \bX \Big( \bX^\top \bX - \Big( \frac kd + 2 \Big) \bI_d \Big) \bX^\top \Big) \circ (\bX \bX^\top) \Big\| 
\le \|\bX\|^2 \Big\| \bX^\top \bX - \Big( \frac kd + 2 \Big) \bI_d \Big\| .
$$
By the concentration of a sample covariance matrix (see, e.g., Theorem~5.7 in \cite{rigollet2023high}), there is an absolute constant $C_1 > 1$ such that for any $\delta \in (0,1)$, with probability at least $1-\delta$,
$$
\Big\| \frac 1k \bX^\top \bX - \frac 1d \bI_d \Big\| \le \frac{C_1}d \Big( \sqrt{\frac{d + \log(1/\delta)}{k}} + \frac{d + \log(1/\delta)}{k} \Big) .
$$
Since $k \ge d + \log(1/\delta)$, we obtain 
$$
\Big\| \bX^\top \bX - \Big( \frac kd + 2 \Big) \bI_d \Big\| \le 2 + 2 C_1 \frac kd \sqrt{\frac{d + \log(1/\delta)}{k}} \le 4 C_1 \frac kd \sqrt{\frac{d + \log(1/\delta)}{k}} 
$$
and
\begin{equation}
\|\bX\|^2 = \|\bX^\top \bX\| \le \frac kd + 2 C_1 \frac kd \sqrt{\frac{d + \log(1/\delta)}{k}} \le (2 C_1 + 1) \frac kd . \label{eq:XTX-norm}
\end{equation}
We conclude that
$$
\Big\| \Big( \bX \Big( \bX^\top \bX - \Big( \frac kd + 2 \Big) \bI_d \Big) \bX^\top \Big) \circ (\bX \bX^\top) \Big\| 
\le 4 C_1 (2C_1 + 1) \frac{k^2}{d^2} \sqrt{\frac{d + \log(1/\delta)}{k}} .
$$
Moreover, the second term in \eqref{eq:K-three-terms} can be bounded as
$$
\|(\bX \bX^\top)^2 \circ \bI_k\|
= \max_{i \in [k]} X_i^\top \bX^\top \bX X_i
\le \|\bX^\top \bX\| \le (2 C_1 + 1) \frac kd .
$$
Combining the above bounds with \eqref{eq:K-three-terms} using the triangle inequality finishes the proof.
\end{proof}

Next, we study the matrix $(\bX \bX^\top) \circ (\bX \bX^\top)$ and justify \eqref{eq:K-K-spherical-harmonics} by choosing suitable $\bY$ and $\bLam$.
Towards this end, we use the tools of spherical harmonics (see \cite{dai2013approximation} for an introduction to this topic). 
Let 
\begin{equation}
D:=(d-2)/2 ,
\label{eq:def-D}
\end{equation}
and consider the $0$th and $2$nd order Gegenbauer polynomials $C_0^D(t)=1$ and $C_2^D(t)=2D(D+1)t^2-D$ respectively. Then we can write
$$t^2=\frac{1}{2(D+1)}C_0^D(t)+\frac{1}{2D(D+1)}C_2^D(t)\,.$$
Moreover, let $H_\ell^d$ denote the space of real harmonic polynomials of degree $\ell$ on $\R^d$. By Corollary~1.1.4 in \cite{dai2013approximation}, $\dim H_0^d = 1$ and 
\begin{equation}
m := \dim H_2^d = \binom{d+1}{2} - 1 . 
\label{eq:def-m}
\end{equation}
Let $\phi_0 \equiv 1$, and let $\phi_1, \dots, \phi_m$ be an orthonormal basis of $H_2^d$. The addition formula for spherical harmonics, (1.2.8) in \cite{dai2013approximation}, states that
\begin{equation}
\frac{D+2}{D} C^D_2(X_i^\top X_j) = \sum_{\ell=1}^m \phi_\ell(X_i) \phi_\ell(X_j) . \label{eq:addition-formula}
\end{equation}
Putting it together, we obtain
\begin{align*}
(X_i^\top X_j)^2 &= \frac{1}{2(D+1)}C_0^D(X_i^\top X_j)+\frac{1}{2D(D+1)}C_2^D(X_i^\top X_j) \\
&= \frac{1}{2(D+1)} \phi_0(X_i) \phi_0(X_j) + \frac{1}{2(D+1)(D+2)} \sum_{\ell=1}^m \phi_\ell(X_i) \phi_\ell(X_j) .
\end{align*}
This can be rewritten as the following result.

\begin{lemma} \label{lem:Y-Lambda-eq}
Let $\bY \in \R^{k \times (m+1)}$ be defined by $\bY_{ij} = \frac{1}{\sqrt{k}} \phi_{j-1}(X_i)$. Let the diagonal matrix $\bLam \in \R^{(m+1) \times (m+1)}$ be defined by $\bLam_{11} = \frac{1}{2(D+1)}$ and $\bLam_{ii} = \frac{1}{2(D+1)(D+2)}$ for $2 \le i \le m+1$. Then \eqref{eq:K-K-spherical-harmonics} holds.
\end{lemma}

Note that the first column of $\bY$ is $\frac{1}{\sqrt{k}} \bsone_k$. Therefore, if $k \ge m+1$, we have the QR decomposition $\bY = \bQ\bR$ such that $\bQ \in \R^{k \times (m+1)}$ has orthonormal columns with the first column also equal to $\frac{1}{\sqrt{k}} \bsone_k$, and $\bR \in \R^{(m+1) \times (m+1)}$ is an upper triangular matrix. The following holds. 

\begin{lemma} \label{lem:QR-decomposition}
Assume $k \ge m+1$ and let $\bQ$ be as defined above.
The top two eigenvalues of the matrix $\bQ \bLam \bQ^\top$ are $\bLam_{11} = \frac{1}{2(D+1)}$ and $\bLam_{22} = \frac{1}{2(D+1)(D+2)}$ respectively, and the top eigenvector is $\frac{1}{\sqrt{k}} \bsone_k$. Moreover, for any $\delta \in (0,1)$, if $k^3 \ge (2D+1)(D+2) \log \frac{2(m+1)}\delta$, then with probability at least $1-\delta$,
$$
\| \bY \bLam \bY^\top - \bQ \bLam \bQ^\top \| \le \frac{1}{2(D+1)(D+2)} \sqrt{\frac{(2D+1)(D+2) \log \frac{2(m+1)}\delta}{k}} .
$$
\end{lemma}

\begin{proof}
The first statement holds by the definition of $\bQ$. To bound the norm, we note that
$$
\|\bY \bLam \bY^\top - \bQ \bLam \bQ^\top\|
= \|\bQ \bR \bLam \bR^\top \bQ^\top - \bQ \bLam \bQ^\top\|
= \| \bR \bLam \bR^\top - \bLam\| .
$$
In the QR decomposition $\bY = \bQ\bR$, the first columns of $\bY$ and $\bQ$ are the same and $\bR$ is upper triangular, so the first column of $\bR$ is $e_1$, the first standard basis vector in $\R^{m+1}$. Moreover, $\bLam_{22} = \cdots = \bLam_{m+1,m+1}$ from Lemma~\ref{lem:Y-Lambda-eq}. 
Therefore,
\begin{align*}
\bR \bLam \bR^\top - \bLam 
&= (\bLam_{11} - \bLam_{22}) e_1 e_1^\top + \bLam_{22} \bR \bR^\top - (\bLam_{11} - \bLam_{22}) e_1 e_1^\top - \bLam_{22} \bI_{m+1} \\
&= \bLam_{22} (\bR \bR^\top - \bI_{m+1}) .
\end{align*}
In addition, $\bY^\top \bY = \bR^\top \bQ^\top \bQ \bR = \bR^\top \bR$. 
As a result, 
$$
\|\bY \bLam \bY^\top - \bQ \bLam \bQ^\top\|
= \bLam_{22} \|\bR \bR^\top - \bI_{m+1}\|
= \bLam_{22} \|\bR^\top \bR - \bI_{m+1}\|
= \bLam_{22} \|\bY^\top \bY - \bI_{m+1}\| .
$$

Lemma~12 in \cite{de2020adaptive} shows that if $\delta \in (0,1)$ and $k^3 \ge \rho(m) \log \frac{2(m+1)}\delta$ where $\rho(m) := \max\{1, \|\sum_{\ell=0}^m \phi_\ell^2\|_\infty - 1\}$, then with probability at least $1 - \delta$,
$$
\|\bY^\top \bY - \bI_{m+1}\| \le \sqrt{\frac{\rho(m) \log \frac{2(m+1)}\delta}{k}} .
$$
By the addition formula \eqref{eq:addition-formula}, for any $x \in \cS^{d-1}$, 
$$
\sum_{\ell=0}^m \phi_\ell(x)^2 - 1
= \sum_{\ell=1}^m \phi_\ell(x)^2
= \frac{D+2}{D} C^D_2(1) 
= (2D+1)(D+2) .
$$
Combining everything finishes the proof.
\end{proof}


\subsection{Bounding the noise terms} \label{sec:noise-terms}

We now turn to the remaining terms in \eqref{eq:all-terms-spectral-norm}. Each off-diagonal entry of the noise matrix $\bW$ is obtained from centering a $\Ber(p)$ or $\Ber(p(\bK_{ij}+1))$ variable, so it is obvious that $\bW$ satisfies the following.

\begin{lemma} \label{lem:W-properties}
For the noise matrix $\bW \in \R^{n \times n}$ defined by $\bW_{ii} = 0$ for $i \in [n]$, $\bW_{ij} = \bA_{ij} - p (\bK_{ij} + 1)$ for distinct $i,j \in [k]$, and $\bW_{ij} = \bA_{ij} - p$ otherwise, we have $\E[\bW \mid \bX] = 0$, $\Var(\bW_{ij} \mid \bX) \le 2p$, and $\E[\bW_{ij}^4 \mid \bX] \le 2p$ for any $i,j \in [n]$.
\end{lemma}

We are ready to bound the noise terms in \eqref{eq:all-terms-spectral-norm}.

\begin{lemma} \label{lem:K2W-norm}
There is an absolute constant $C>0$ such that the following holds for any $\delta \in (0, 1)$. 
If $k \ge d + \log(2/\delta)$ and $kp \geq C\log(2k/\delta)$, then with probability at least $1 - \delta$, 
$$
\|(\bK^2)^\# \circ \bW\| \le C \frac{k^{3/2}\sqrt{p}}{d}.
$$
\end{lemma}

\begin{proof}
Let $\bW_{[k]}$ denote the top-left $k \times k$ principal minor of $\bW$. By Lemma~\ref{lem:hadamard-spectral-norm}, 
$$
\|(\bK^2)^\# \circ \bW\|
= \|\bK^2 \circ \bW_{[k]}\|
\le \max_{i \in [k]} (\bK^2)_{ii} \cdot \|\bW_{[k]}\| .
$$
We have
$$
(\bK^2)_{ii} = \sum_{j=1}^k \bK_{ij} \bK_{ji} 
= X_i^\top \sum_{j \ne i} X_j X_j^\top X_i 
= X_i^\top \bX^\top \bX X_i - 1 .
$$
By \eqref{eq:XTX-norm}, for $k \ge d + \log(2/\delta)$ and an absolute constant $C_1>0$,
\begin{equation} 
\max_{i \in [k]} (\bK^2)_{ii} \le \|\bX^\top \bX\| + 1 \le C_1 k/d \label{eq:K2ii-max}
\end{equation}
with probability at least $1-\delta/2$.
Moreover, by Theorem~\ref{thm:BvH} and Lemma~\ref{lem:W-properties}, for an absolute constant $C_2>0$, we claim that
\begin{equation} 
\|\bW_{[k]}\| \le C_2 \sqrt{kp + \log(k/\delta)} \label{eq:W_k-norm}
\end{equation}
with probability at least $1-\delta/2$ as long as $kp \geq C\log(2k/\delta)$. Indeed, conditionally on $\bX$, Theorem~\ref{thm:BvH} applies with $\sigma \leq \sqrt{2kp}$ and $\sigma_\ast \leq 1$. It gives that 
\[
    \mathbb{P}(\|\bW_{[k]}\| \geq C_2 \sqrt{kp + \log(k/\delta)}) \leq k\exp(-t^2/c_{\epsilon}),
\]
as long as $C_2 \sqrt{kp+\log(k/\delta)} \geq 4(1+\epsilon)\sqrt{kp}+t$, since the right-hand side is at least $(1+\epsilon)2\sqrt{2}\sigma + t$. We satisfy this by choosing $t = C_2 \sqrt{kp+\log(k/\delta)} - 4(1+\epsilon)\sqrt{kp}$. Thus the desired probability is upper-bounded by $\delta/2$ if we have $t \geq \sqrt{c_\epsilon\log(2k/\delta)}$. If we choose, say, $\epsilon = 1/2$ and $C_2 = 4(1+\epsilon)+1$, then $t \geq \sqrt{kp}$, which is at least $\sqrt{C\log(2k/\delta)}$ by assumption, as long as we choose the absolute constant to satisfy $C \geq c_{1/2}$. Our assumptions imply that $C_2\sqrt{pk+\log(k/\delta)} \leq C\sqrt{kp}$ for some absolute $C$, which we combine with \eqref{eq:K2ii-max} to complete the proof. 
\end{proof}

\begin{lemma} \label{lem:W2K-norm}
There is an absolute constant $C>0$ such that for any $\delta \in (0, 1)$, if $np \geq C\log(2n/\delta)$ then it holds with probability at least $1 - \delta$ that
$$
\|\bW^2 \circ \bK^\#\| \le C \Big( np + \log \frac n\delta \Big) .
$$
\end{lemma}

\begin{proof}
Let $\bW_{1:k} \in \R^{n \times k}$ denote the matrix consisting of the first $k$ columns of $\bW$. By Lemma~\ref{lem:hadamard-spectral-norm}, 
\begin{align*}
\|\bW^2 \circ \bK^\#\| = \|(\bW_{1:k}^\top \bW_{1:k}) \circ \bK\| 
&\le \|(\bW_{1:k}^\top \bW_{1:k}) \circ (\bX\bX^\top)\| + \|(\bW_{1:k}^\top \bW_{1:k}) \circ \bI_k\| \\
&\le 2 \|\bW_{1:k}^\top \bW_{1:k}\| 
\le 2 \|\bW\|^2 .
\end{align*}
Moreover, by Theorem~\ref{thm:BvH} and Lemma~\ref{lem:W-properties}, for an absolute constant $C_1>0$,
\begin{equation}
\|\bW\| \le C_1 \sqrt{np + \log(n/\delta)} \label{eq:W-norm}
\end{equation}
with probability at least $1-\delta$, as long as $np \geq C\log(2n/\delta)$. The proof of this goes as in \eqref{eq:W_k-norm}; since $\bW$ is $n \times n$ whereas $\bW_{[k]}$ is $k \times k$, we just replace $k$ with $n$ everywhere in the argument. The result then follows.
\end{proof}

\begin{lemma} \label{lem:KWWKK-norm}
There is an absolute constant $C>0$ such that the following holds for any $\delta \in (0, 1)$. 
If $k \ge d + \log(2/\delta)$, then with probability at least $1 - \delta$, 
$$
\|(\bK^\# \bW + \bW \bK^\#) \circ \bK^\#\| \le C \frac{k}{d} \sqrt{kp + \log \frac k\delta} .
$$
\end{lemma}

\begin{proof}
Let $\bW_{[k]}$ denote the top-left $k \times k$ principal minor of $\bW$.
Similar to the beginning of the proof of Lemma~\ref{lem:W2K-norm}, it holds that
\begin{align*}
\|(\bK^\# \bW + \bW \bK^\#) \circ \bK^\#\| &= \|(\bK \bW_{[k]} + \bW_{[k]} \bK) \circ \bK\| \\
&\le 2 \|\bK \bW_{[k]} + \bW_{[k]} \bK\| 
\le 4 \|\bK\| \|\bW_{[k]}\| .
\end{align*}
By \eqref{eq:XTX-norm}, for $k \ge d + \log(2/\delta)$ and an absolute constant $C_1>0$,
$$
\|\bK\| \le \|\bX^\top \bX\| + 1 \le C_1 k/d
$$
with probability at least $1-\delta/2$. 
This together with \eqref{eq:W_k-norm} finishes the proof.
\end{proof}

To study the term $(\bK^\# \bW + \bW \bK^\#) \circ \bW$, we first apply a decoupling inequality.

\begin{lemma}\label{lemma:newdecoupling}
Let $\bW'$ be an independent copy of $\bW$ conditional on $\bX$. 
There exists an absolute constant $C>0$ such that for all $t>0$, 
$$\bbP\left\{\norm{(\bK^\# \bW + \bW \bK^\#) \circ \bW}>t \mid \bX\right\}\leq C\bbP\left\{C\norm{(\bK^\# \bW + \bW \bK^\#) \circ \bW'}>t \mid \bX\right\}\,.$$
\end{lemma}

\begin{proof}
Let $\Omega=[-1,1]$ and let $\cB$ denote the separable Banach space $\R^{n \times n}$ with the spectral norm. 
Let
\[
    J_n:=\big\{((i,j'),(j,\ell))\in([n]^2)^2:i<j', \text{ and } j<\ell, \text{ and }(i,j')\neq(j,\ell)\big\}.
\]
For $((i,j'),(j,\ell)) \in J_n$, define the function $h_{ij',j\ell} : \Omega^2 \to \cB$ by 
$$h_{ij',j\ell}(x,y):=xy\cdot I_{ij'j\ell}\cdot(e_ie_{j'}^\top+e_{j'}e_i^\top)\,,$$
where
$$I_{ij'j\ell}:=\bsone\{i=j\}\,\bK^\#_{j'\ell}+\bsone\{i=\ell\}\,\bK^\#_{jj'}+\bsone\{j=j'\}\,\bK^\#_{i\ell}+\bsone\{j'=\ell\}\,\bK^\#_{ij}$$
and $e_i$ is the $i$th standard basis vector. One then has
$$(\bK^\# \bW + \bW \bK^\#) \circ \bW = \sum_{((i,j'),(j,\ell)) \in J_n} h_{ij',j\ell}(\bW_{ij'},\bW_{j\ell})\,,$$
and one can similarly check that
\[
    (\bK^\# \bW + \bW \bK^\#) \circ \bW' = \sum_{((i,j'),(j,\ell)) \in J_n} h_{ij',j\ell} (\bW'_{ij'},\bW_{j\ell})\,,
\]
so the result follows by applying Theorem~3.4.1 in \cite{de2012decoupling}.
\end{proof}

\begin{lemma} \label{lem:KWWKW-norm}
There is an absolute constant $C>0$ such that if $\delta \in (0,1)$ satisfies $n p \ge \log(n/\delta)$, then with probability at least $1 - \delta$, 
$$
\| (\bK^\# \bW + \bW \bK^\#) \circ \bW \| \le C \sqrt{\frac{np \log(n/\delta)}{d}} \left( \sqrt{k p \log(n/\delta)} + \log(n/\delta) \right)  .
$$
\end{lemma}

\begin{proof}
First, recall that $\bK^\#_{ij} = \bK_{ij} = X_i^\top X_j$ for distinct $i,j \in [k]$ and $\bK^\#_{ij} = 0$ otherwise. 
Since $X_i^\top X_j$ is sub-Gaussian with variance $O(1/d)$, there is an absolute constant $C_1>0$ such that with probability at least $1-\delta$,
$$
|\bK_{ij}| \le C_1 \sqrt{\frac{\log(n/\delta)}{d}}
$$
for all $i,j \in [n]$. Let us condition on an instance of $\bX$ such that the above bound holds.

Moreover, recall that $\bW_{ij} \in [-1,1]$, $\E[\bW \mid \bX] = 0$, and $\Var(\bW_{ij} \mid \bX) \le 2p$. 
By Bernstein's inequality, there is an absolute constant $C_2>0$ such that with probability at least $1-\delta$, 
$$
|(\bK^\# \bW)_{ij}| 
= \left|\sum_{\ell = 1}^k \bK_{i\ell} \bW_{\ell j}\right|
\le C_2 \sqrt{\frac{\log(n/\delta)}{d}} \left( \sqrt{k p \log(n/\delta)} + \log(n/\delta) \right)
$$
for all $i,j \in [n]$. We further condition on an instance of $\bW$ such that the above bound holds. 

Let $\bW'$ be an independent copy of $\bW$ conditional on $\bX$. Then for an absolute constant $C_3 > 0$, we have
$$
\E[((\bK^\# \bW + \bW \bK^\#) \circ \bW')_{ij}^2 \mid \bX, \bW] 
\le C_3 p \frac{\log(n/\delta)}{d} \left( k p \log(n/\delta) + (\log(n/\delta))^2 \right)
$$
and 
$$
|((\bK^\# \bW + \bW \bK^\#) \circ \bW')_{ij}| \le C_3 \sqrt{\frac{\log(n/\delta)}{d}} \left( \sqrt{k p \log(n/\delta)} + \log(n/\delta) \right) 
$$
for all $i,j \in [n]$. 
By Theorem~\ref{thm:BvH}, there are absolute constants $C_4,C_5>0$ such that with probability at least $1-\delta$,
\begin{align*}
\norm{(\bK^\# \bW + \bW \bK^\#) \circ \bW'} &\le C_4 \bigg( \sqrt{n p \frac{\log(n/\delta)}{d} \left( k p \log(n/\delta) + (\log(n/\delta))^2 \right)} \\
&\qquad + \sqrt{\frac{\log(n/\delta)}{d}} \left( \sqrt{k p \log(n/\delta)} + \log(n/\delta) \right) \sqrt{\log(n/\delta)} \bigg) \\
&\le C_5 \sqrt{\frac{\log(n/\delta)}{d}} \left( \sqrt{k p \log(n/\delta)} + \log(n/\delta) \right) \sqrt{np + \log(n/\delta)}
\end{align*}
Combining the above with Lemma~\ref{lemma:newdecoupling} completes the proof.
\end{proof}

It remains to control $\| \bW^2 \circ \bW \|$. The following bound is an immediate consequence of Theorem~\ref{thm:MainNoiseTerm} together with the properties of $\bW$ in Lemma~\ref{lem:W-properties}.

\begin{lemma} \label{lem:W2W-norm}
There is an absolute constant $C>0$ such that if $\delta \in (0,1)$ satisfies $n \ge \log^3(n/\delta)$ and $n p \ge \log^{5/3}(n/\delta)$, then with probability at least $1-\delta$, 
$$
\| \bW^2 \circ \bW \| \le C \left( np^{3/2} + n^{1/2}p\log(n/\delta) + \log^{3/2}(n/\delta) \right).
$$
\end{lemma}


\subsection{Proof of Proposition~\ref{prop:main-spectral}}
Recall the definition of $D$ in \eqref{eq:def-D} and the definitions of $\bLam_{11}$ and $\bLam_{22}$ in Lemma~\ref{lem:Y-Lambda-eq}. The top two eigenvalues and the top eigenvector of $\widetilde{\bM}$ are given by \Cref{lem:QR-decomposition} in view of the definition of $D$ in \eqref{eq:def-D}.

To control the spectral norm, we bound the right-hand side of \eqref{eq:all-terms-spectral-norm} with probability at least $1-n^{-10}$. 
Towards this end, we apply Lemmas~\ref{lem:K2-K-K-K}, \ref{lem:QR-decomposition}, \ref{lem:K2W-norm}, \ref{lem:W2K-norm}, \ref{lem:KWWKK-norm}, \ref{lem:KWWKW-norm}, and \ref{lem:W2W-norm} with $\delta = n^{-10}/7$. 
Also recall the definitions \eqref{eq:def-D} and \eqref{eq:def-m}. 
To see that these lemmas are applicable, we first check that the conditions assumed in the lemmas are all satisfied thanks to the assumption $k \ge C \sqrt{n} d / p^{3/4}$:
\begin{itemize}
\item
$k \ge d + \log(2/\delta)$: Obvious in view of the condition $k \ge C \sqrt{n} d / p^{3/4} \ge C(d + \log n)$.

\item
$k \ge m+1$: It suffices to have $k \ge d^2$ which is subsumed by $k \ge C \sqrt{n} d / p^{3/4}$ if $d \le \sqrt{n}$. Since $k \le n$, we do have $\sqrt{n} \ge C d/p^{3/4} \ge d$.

\item
$k^3 \ge (2D+1)(D+2) \log \frac{2(m+1)}\delta$: It suffices to have $k^3 \ge C d^2 \log n$ which is subsumed by $k \ge d^2$ proved above together with $k \ge C \sqrt{n} d / p^{3/4} \ge C \log n$.

\item
$kp \geq C\log(2k/\delta)$: Since $k \leq n$, it suffices to show that $kp \geq C\log(14n^{11})$. As noted just after Theorem \ref{thm:main-spectral}, our conditions imply $np^{3/2} \geq C$, so $kp = kp^{3/4} p^{1/4} \geq Cn^{1/2}d \cdot C^{1/6} n^{-1/6} \geq C^{7/6} n^{1/3}$ which suffices.

\item
$np \geq C\log(2n/\delta)$: The previous bullet point actually showed that $kp \geq C\log(2n/\delta)$, which is stronger.

\item
$np \geq \log(n/\delta)$: Weaker than the previous bullet point.

\item
$n p \ge \log^{5/3}(n/\delta)$: 
We showed above that, under the main assumption, we have $np^{3/2} \geq C$, so that $p \geq Cn^{-2/3}$, so $np \geq Cn^{1/3} \geq \log^{5/3}(n/\delta)$.

\item
$n \ge \log^3(n/\delta)$: Obvious.
\end{itemize}
Now we can apply the above lemmas together with \eqref{eq:all-terms-spectral-norm} to obtain 
\begin{align*}
\left\|\bM - \widetilde{\bM}\right\|
&\le C_1 \bigg( \frac{p^3 k^2}{d^2} \sqrt{\frac{d + \log(1/\delta)}{k}} + \frac{p^3 k^2}{d^2} \sqrt{\frac{\log (d/\delta)}{k}} \\
& + \frac{p^{5/2} k^{3/2}}{d} + np^2 + p \log \frac n\delta + np^{3/2} + n^{1/2}p\log(n/\delta) + \log^{3/2}(n/\delta) \\
& + \frac{p^2 k}{d} \sqrt{kp + \log \frac k\delta} + p \sqrt{\frac{np}{d}} \left( \sqrt{k p} \log(n/\delta) + \log^{3/2}(n/\delta) \right) \bigg) 
\end{align*}
for an absolute constant $C_1>0$ with probability at least $1 - \delta$ where $\delta = n^{-10}/7$. 

In addition, $kp \ge \log n$ (an intermediate result of the fourth bullet point above), so the above bound simplifies to 
\begin{align*}
&\left\|\bM - \widetilde{\bM}\right\| \\
&\le C_2 \bigg( \frac{p^3 k^{3/2}}{d^2} \sqrt{\log n} 
+ \frac{p^{5/2} k^{3/2}}{d} + np^{3/2} + n^{1/2}p\log n + \log^{3/2} n 
+ p^2 \sqrt{\frac{k n}{d}} \log n\bigg) 
\end{align*}
for an absolute constant $C_2>0$. 
Finally, it remains to use the assumptions $k \ge C \sqrt{n} d / p^{3/4}$ and $n p^{3/2} \ge \log n$ to show that each of the following quantities is smaller than $\eps/(5 C_2)$ if $C$ is sufficiently large: 
\begin{itemize}
\item
$\frac{d^2}{p^3 k^2} \frac{p^3 k^{3/2}}{d^2} \sqrt{\log n}$: It suffices to have $k \ge C \log n$ which is obvious.
\item
$\frac{d^2}{p^3 k^2} \frac{p^{5/2} k^{3/2}}{d}$: It suffices to have $p k \ge C d^2$. Since $n \ge k \ge C \sqrt{n} d / p^{3/4}$, we have $\sqrt{np} \ge C d/p^{1/4}$, and $kp \ge C \sqrt{n} d p^{1/4} = C n^{1/4} d (np)^{1/4} \ge C n^{1/4} d^{3/2}/p^{1/8}$. Since we showed that $d \le \sqrt{n}$ (see the second bullet point above), it follows that $kp \ge C d^2 / p^{1/8} \ge C d^2$.
\item
$\frac{d^2}{p^3 k^2} np^{3/2}$: The assumption $k \ge C \sqrt{n} d / p^{3/4}$ is sufficient.
\item
$\frac{d^2}{p^3 k^2} n^{1/2}p\log n$: Since $\frac{d}{k} \leq \frac{p^{3/4}}{C\sqrt{n}}$, we have
\[
    \frac{d^2}{p^3k^2} n^{1/2}p\log(n) \leq \frac{n^{1/2}\log(n)}{p^2} \cdot \frac{p^{3/2}}{C^2 n} = \frac{\log(n)}{C^2 (np)^{1/2}} = \frac{1}{C^2} \cdot \left(\frac{\log^2(n)}{np}\right)^{1/2}.
\]
Previously we showed that, under our main assumption, we have $np^{3/2} \geq C$, so that $p \geq Cn^{-2/3}$. This means that $np \geq Cn^{1/3}$, and thus $\log^2(n)/(np) = o(1)$.
\item
$\frac{d^2}{p^3 k^2} \log^{3/2} n$: It suffices to have 
\[
    k \ge C d (\log n)^{3/4} / p^{3/2}
\]
which is subsumed by $k \ge C \sqrt{n} d / p^{3/4}$ since $n p^{3/2} \ge (\log n)^{3/2}$. 
\item
$\frac{d^2}{p^3 k^2} p^2 \sqrt{\frac{k n}{d}} \log n$: It suffices to have 
$k \ge C d n^{1/3} (\log n)^{2/3} / p^{2/3}$ which is subsumed by $k \ge C \sqrt{n} d / p^{3/4}$.
\end{itemize}
This completes the proof.


\section{Spectral norm of a graph matrix} \label{sec:matrix-chaos-bound}
In this section we establish \Cref{thm:MainNoiseTerm}, which immediately implies Lemma~\ref{lem:W2W-norm}. 

\begin{theorem}\label{thm:MainNoiseTerm}
Let $\bZ\in\bbR^{n\times n}$ be a symmetric matrix with independent entries above the diagonal and zeros on the diagonal. For all $1\leq i<j\leq n$, assume $\norm{\bZ_{ij}}_\infty\leq1$, $\bbE \bZ_{ij}=0$, $\bbE \bZ_{ij}^2\leq K$, and $\bbE \bZ_{ij}^4\leq K$ for a quantity $K \in (0,1]$ that may depend on $n$. 
Then there is an absolute constant $C>0$ such that if $\delta \in (0,1)$ satisfies $n \ge \log^3(n/\delta)$ and $n K \ge \log^{5/3}(n/\delta)$, then the following holds. With probability at least $1-\delta$, the matrix $\bL:=\bZ^2\circ \bZ$ satisfies
$$
\norm{\bL} \le C \left( nK^{3/2} + n^{1/2}K\log(n/\delta) + \log^{3/2}(n/\delta) \right) .
$$
\end{theorem}

A direct application of Theorem 4.1 in \cite{rajendran2023concentration} implies the norm of the matrix $\bL$ from \Cref{thm:MainNoiseTerm} is of order $nK^{3/2}$ with an extra multiplicative logarithmic factor. 
The more recent work \cite{bandeira2025matrix} proves a general bound of the desired order on the expected spectral norm but does not directly yield a high-probability bound that we need. Our purpose in this section is to prove a high-probability spectral norm bound without any logarithmic factor in the main term.

The first step toward proving \Cref{thm:MainNoiseTerm} is to decouple $\bL = \bZ^2\circ \bZ$: If $\bZ'$ is an independent copy of $\bZ$, then we will show that it suffices to bound $\|\widetilde{\bL}\|$ where $\widetilde{\bL}=\bZ^2\circ \bZ'$. The benefit of this decoupling---which is the main idea of the proof of Theorem \ref{thm:MainNoiseTerm}---is that, if we decouple and then condition on $\bZ$, then the matrix $\widetilde{\bL}$ has (conditionally) independent entries up to symmetry, so we can apply Theorem~\ref{thm:BvH}.

\begin{lemma}\label{lemma:decoupling}
Let $\bZ$ be as defined in Theorem~\ref{thm:MainNoiseTerm}.
Let $\bZ'$ be an independent copy of $\bZ$. 
Define the matrices $\bL=\bZ^2\circ \bZ$ and $\widetilde{\bL}=\bZ^2\circ \bZ'$. There exists an absolute constant $C>0$ such that for all $t>0$,
$$\bbP\left\{\norm{\bL}>t\right\}\leq C\bbP\{C\norm{\widetilde{\bL}}>t\}\,.$$
\end{lemma}

\begin{proof}
We will prove this result by applying the decoupling inequalities of \cite{de2012decoupling}, which means that we need to find a way to write $\bL$ and $\widetilde{\bL}$ in the form described there, namely as functions of some underlying independent random variables. The entries of $\bZ$ are not independent, because of the symmetry. However, if we (abusively) write $\bZ_{\{i,j\}}$ for the common value $\bZ_{ij} = \bZ_{ji}$ when $\{i,j\} \in \binom{[n]}{2}$ (i.e., $\{i,j\}$ is an unordered pair with $i \neq j$), then the random variables $\bZ_{\{i,j\}}$ are independent as $\{i,j\}$ ranges over $\binom{[n]}{2}$. They take values in $\Omega = [-1,1]$. Let $J_n \subset \binom{[n]}{2}^3$ be the set of ordered triples of \emph{distinct} elements of $\binom{[n]}{2}$. For each such triple $(\{i,j'\},\{j,k'\},\{k,i'\})$, define the function $h_{\{i,j'\},\{j,k'\},\{k,i'\}} : \Omega^3 \to \cB$, where $\cB$ is the separable Banach space of real-symmetric $n \times n$ matrices equipped with the spectral norm, by
\[
    h_{\{i,j'\},\{j,k'\},\{k,i'\}}(x,y,z) = \frac{1}{2} xyz \cdot \bsone\{(\{i,j'\}, \{j, k'\}, \{k, i'\}) \text{ form a triangle}\} \cdot (e_i e_{j'}^\top + e_{j'} e_i^\top),
\]
where $e_i$ is the $i$th standard basis vector, and where we say that $(\{i,j'\},\{j,k'\},\{k,i'\})$ form a triangle if all pairs of unordered pairs have nonempty intersection (i.e., if the triple is really of the form $(\{a,b\},\{b,c\},\{a,c\})$ up to reordering). Let $(\bZ_{\{i,j\}}^{(a)})_{\{i,j\} \in\binom{[n]}{2}}$ for $a = 1,2,3$ be independent copies of $(\bZ_{\{i,j\}})_{\{i,j\} \in\binom{[n]}{2}}$. For $a,b,c\in\{1,2,3\}$, write
\[
    S_{abc}
    :=
    \sum_{(\{i,j'\},\{j,k'\},\{k,i'\})\in J_n}
    h_{\{i,j'\},\{j,k'\},\{k,i'\}}
    (\bZ^{(a)}_{ij'},\bZ^{(b)}_{jk'},\bZ^{(c)}_{ki'}).
\]
Then $S_{111}$ has the same law as $\bL$, whereas $S_{122}$ has the same law as $\widetilde{\bL}=\bZ^2\circ\bZ'$.
By the decoupling inequality in \cite[Theorem~3.4.1]{de2012decoupling}, applied to the independent variables indexed by $\binom{[n]}{2}$, there is a universal constant $C_1>0$ such that
\[
    \bbP\{\|S_{111}\|>t\}\le C_1\bbP\{C_1\|S_{123}\|>t\}.
\]
Next, conditional on $(\bZ_{\{i,j\}}^{(1)})_{\{i,j\} \in\binom{[n]}{2}}$, the sum $S_{123}$ is an order-two decoupled $U$-statistic in the second and third coordinates. In other words, we can write it as
\[
    S_{123} = \sum_{(\{j,k'\},\{k,i'\}) \in \widetilde{J_n}} \widetilde{h}_{\{j,k'\},\{k,i'\}}(\bZ^{(2)}_{jk'},\bZ^{(3)}_{ki'}),
\]
where $\widetilde{J_n} \subset \binom{[n]}{2}^2$ is the set of ordered pairs of distinct elements of $\binom{[n]}{2}$ and
\begin{align*}
    &\widetilde{h}_{\{j,k'\},\{k,i\}}(y,z) \\
    &= \frac{yz}{2} \left( \sum_{\{i,j'\} \in \binom{[n]}{2} \setminus \{\{j,k'\},\{k,i\}\}} \bsone\{(\{i,j'\}, \{j, k'\}, \{k, i'\}) \text{ form a triangle}\} \bZ^{(1)}_{ij'} (e_i e_{j'}^\top + e_{j'} e_i^{\top}) \right).
\end{align*}
Notice that these kernels are symmetric in the sense that $\widetilde{h}_{p_1,p_2}(y,z) = \widetilde{h}_{p_2,p_1}(z,y)$. Hence the reverse decoupling inequality in \cite[Theorem~3.4.1]{de2012decoupling}, applied conditionally on $(\bZ_{\{i,j\}}^{(1)})_{\{i,j\} \in\binom{[n]}{2}}$, gives a universal constant $C_2>0$ such that
\[
    \bbP\{C_1\|S_{123}\|>t\}
    \le C_2\bbP\{C_1C_2\|S_{122}\|>t\}.
\]
Combining the above two inequalities gives the desired result.
\end{proof}

As previously mentioned, the crux of the proof of \Cref{thm:MainNoiseTerm} is an application of Theorem~\ref{thm:BvH} of Bandeira and Van Handel to the partially decoupled matrix $\widetilde{\bL}$. 
To apply \Cref{thm:BvH}, we need good control of the variables $\sigma$ and $\sigma_\ast$, which is given by the following two lemmas.

\begin{lemma}\label{lemma:sigmaTail}
Assume $\bZ\in\bbR^{n\times n}$ satisfies the hypotheses of \Cref{thm:MainNoiseTerm}. For all $i \in [n]$, define
$$\sigma_i := \sqrt{\sum_{j\in [n] \setminus \{i\}} (\bZ^2)_{ij}^2}\,,$$
where, for a matrix $\bX\in\bbR^{n\times n}$, we write $\bX^2_{ij}=(\bX_{ij})^2$.
Then there is an absolute constant $C>0$ such that if $\delta \in (0,1)$ satisfies $n K \ge \log(n/\delta)$ and $n \ge \log^3(n/\delta)$, then the following holds. With probability at least $1-\delta$, we have that for all $i \in [n]$,
$$
\sigma_i \le C \left( n K + n^{1/4} K^{1/4} \log^{5/4}(n/\delta) \right).
$$
\end{lemma}

\begin{lemma}\label{lemma:sigmastarTail}
Assume $\bZ\in\bbR^{n\times n}$ satisfies the hypotheses of \Cref{thm:MainNoiseTerm}. Then there is an absolute constant $C>0$ such that the following holds for any $\delta \in (0,1)$. With probability at least $1-\delta$, we have that for all $i,j \in [n]$,
$$
|(\bZ^2)_{ij}| \le C \left( n^{1/2} K \log^{1/2}(n/\delta) + \log(n/\delta) \right).
$$
\end{lemma}

We first prove \Cref{thm:MainNoiseTerm} using \Cref{thm:BvH,lemma:sigmaTail,lemma:sigmastarTail}.

\begin{proof}[Proof of \Cref{thm:MainNoiseTerm}]
Let $\bZ'$ be an independent copy of $\bZ$ and let $\widetilde{\bL}:=\bZ^2 \circ \bZ'$. By \Cref{lemma:decoupling}, it suffices to prove the tail bound with $\widetilde{\bL}$ in place of $\bL$. 
In the sequel, we condition on a realization of $\bZ$ such that the bounds in \Cref{lemma:sigmaTail,lemma:sigmastarTail} hold, and all the probabilities and expectations are with respect to $\bZ'$ conditional on $\bZ$. 
Define the quantities
$$\sigma:=\max_{i=1,\dots,n}\sqrt{\sum_{j=1}^n\bbE\Big[\big(\widetilde{\bL}_{ij}\big)^2\Big]}\,,\hspace{7mm}\sigma_\ast:=\max_{i,j=1,\dots,n}\norm{\widetilde{\bL}_{ij}}_\infty\,.$$
Since $\bbE (\bZ'_{ij})^2\leq K$, by the bound in \Cref{lemma:sigmaTail}, we obtain
$$\sigma=\max_{i=1,\dots,n}\sqrt{\sum_{j \in [n] \setminus \{i\}} (\bZ^2)_{ij}^2 \bbE (\bZ'_{ij})^2}\leq C \sqrt{K}\left( n K + n^{1/4} K^{1/4} \log^{5/4}(n/\delta) \right) .$$
Since $\norm{\bZ'_{ij}}_\infty \le 1$, by the bound in \Cref{lemma:sigmastarTail}, we obtain
$$
\sigma_\ast = \max_{i,j=1,\dots,n} |(\bZ^2)_{ij}| \cdot \norm{\bZ'_{ij}}_\infty
\le C \left( n^{1/2} K \log^{1/2}(n/\delta) + \log(n/\delta) \right) .
$$
We then claim that, for any $\delta \in (0,1)$, we have
\begin{align*}
\norm{\widetilde{\bL}}
&\lesssim \sigma + \sigma_\ast \sqrt{\log(n/\delta)} \\
&\lesssim \sqrt{K}\left( n K + n^{1/4} K^{1/4} \log^{5/4}(n/\delta) \right) 
+ \left( n^{1/2} K \log^{1/2}(n/\delta) + \log(n/\delta) \right) \sqrt{\log(n/\delta)} \\
&\lesssim n K^{3/2} + n^{1/2} K \log(n/\delta) + \log^{3/2}(n/\delta)
\end{align*}
with probability at least $1-\delta$. The first and second steps use \Cref{thm:BvH} and the estimates earlier in this proof. In the third step, we use our assumption that $nK \geq \log^{5/3}(n/\delta)$, so that $\log(n/\delta) \leq (nK)^{3/5}$ and thus $\log^{5/4}(n/\delta) \leq (nK)^{3/4}$. This implies that $n^{1/4}K^{3/4}\log^{5/4}(n/\delta) \leq nK^{3/2}$. Taking a union bound completes the proof.
\end{proof}

\begin{proof}[Proof of \Cref{lemma:sigmaTail}]
Fix $i \in [n]$. Let $\bZ_i$ denote the $i$th column of $\bZ$. Note that $$\sigma_i^2=\sum_{j\in[n]\setminus\{i\}}\inner{\bZ_i,\bZ_j}^2\,.$$
If we define the quantities
$$S_1:=\sum_{\substack{j,r\in[n]\setminus\{i\}\\j\neq r}}\bZ_{ir}^2\bZ_{jr}^2\,,\hspace{7mm}S_2:=\sum_{\substack{j,r,s\in[n]\setminus\{i\}\\\text{all distinct}}}\bZ_{ir}\bZ_{jr}\bZ_{is}\bZ_{js}\,,$$
then $\sigma_i^2=S_1+S_2$. 

Using the fact that $\norm{\bZ_{jr}^2}_\infty\leq1$, $\bbE \bZ_{jr}^2\leq K$, and $\Var[\bZ_{jr}^2]\leq K$, Bernstein's inequality implies that for all $r\in[n]\setminus\{i\}$, the event
\begin{equation}
\sum_{j \in [n] \setminus \{i,r\}} \bZ_{jr}^2 \lesssim nK+\sqrt{nK\log\left(\frac{n}{\delta}\right)}+\log\left(\frac{n}{\delta}\right) 
\lesssim nK \label{eqn:BernsteinS1}
\end{equation}
occurs with probability at least $1-\delta/n$, using the assumption $n K \ge \log(n/\delta)$. Hence by applying \pref{eqn:BernsteinS1} $n$ times ($n-1$ times corresponding with the indexes $r\in[n]\setminus\{i\}$ and once to the sum over $r$), we find that the event
\begin{align*}
S_1 = \sum_{r\in[n]\setminus\{i\}}\bZ_{ir}^2\sum_{j\in[n]\setminus\{i,r\}}\bZ_{jr}^2 &\lesssim nK \sum_{r\in[n]\setminus\{i\}}\bZ_{ir}^2 
\lesssim n^2K^2 \label{eqn:sigmaBoundS1}
\end{align*}
occurs with probability at least $1-\delta$, which follows from a union bound.

To bound $S_2$, first define $B_{rs}:=\sum_{j\in[n]\setminus\{i,r,s\}}\bZ_{jr}\bZ_{js}$ for all $r,s\in[n]\setminus\{i\}$ such that $r\neq s$. The quantities $B_{rs}$ satisfy
\begin{equation*}
S_2=\sum_{r\in[n]\setminus\{i\}}\bZ_{ir}\sum_{s\in[n]\setminus\{i,r\}}\bZ_{is}B_{rs}\,.\label{eqn:S2TwoSums}
\end{equation*}
Let $\{\bZ_{ij}'\}_{j\in[n]\setminus\{i\}}$ be an independent copy of the collection $\{\bZ_{ij}\}_{j\in[n]\setminus\{i\}}$, and define
\begin{equation}
S_2':=\sum_{r\in[n]\setminus\{i\}}\bZ_{ir}\sum_{s\in[n]\setminus\{i,r\}}\bZ_{is}'B_{rs}\,.\label{eqn:Sprime2TwoSums}
\end{equation}
For all $(r,s) \in ([n] \setminus \{i\})^2$ such that $r\neq s$, define $h_{r,s}:\bbR^2\to\bbR$ by
$$h_{r,s}(x,y):=xy\cdot B_{rs}\,.$$
Since
$$S_2=\sum_{\substack{(r,s) \in ([n] \setminus \{i\})^2\\r\neq s}}h_{r,s}(\bZ_{ir},\bZ_{is})\,,$$
we can apply the decoupling inequality in Theorem~3.4.1 from \cite{de2012decoupling} to deduce that for all $t>0$, 
\begin{equation}
\bbP\left\{|S_2|\geq t\right\}\leq C\cdot\bbP\left\{C\cdot |S_2'|\geq t\right\}\label{eqn:S2NullDecoupling}
\end{equation}
for an absolute constant $C>0$. (Notice that, although the functions $h_{r,s}$ are random, the collection $(B_{rs})_{r,s \in [n] \setminus \{i\}}$ is independent of the collection $(\bZ_{ir})_{r \in ([n] \setminus \{i\})}$, because the former sees neither the $i$th row nor the $i$th column of $\bZ$. Thus we can first condition on $B_{rs}$, retaining the independence of the $\bZ_{ir}$'s necessary to apply \cite{de2012decoupling}, then use that $C$ is universal to see that the bound also holds unconditionally.)

We will apply Bernstein's inequality three times iteratively to the terms in \pref{eqn:Sprime2TwoSums}. First, since $\norm{\bZ_{jr}\bZ_{js}}_\infty\leq1$ and $\Var[\bZ_{jr}\bZ_{js}]\leq K^2$, Bernstein's inequality implies that for all $r,s\in[n]\setminus\{i\}$ with $r\neq s$, the event
\begin{equation}
|B_{rs}|\lesssim\sqrt{nK^2\log\left(\frac{n}{\delta}\right)}+\log\left(\frac{n}{\delta}\right)\label{eqn:BklUpperbd}
\end{equation}
occurs with probability at least $1-\delta/n^2$. For all $r\in[n]\setminus\{i\}$, define
$$D_r:=\sum_{s\in[n]\setminus\{i,r\}}\bZ_{is}'B_{rs}\,.$$ 
For fixed values of $B_{rs}$ for $s\in[n]\setminus\{i,r\}$, Bernstein's inequality implies that the event
\begin{equation}
|D_r|=\left|\sum_{s\in[n]\setminus\{i,r\}}\bZ_{is}'B_{rs}\right|\lesssim\sqrt{\sum_{s\in[n]\setminus\{i,r\}}KB_{rs}^2\log\left(\frac{n}{\delta}\right)}+\max_s|B_{rs}|\log\left(\frac{n}{\delta}\right)\label{eqn:CkUpperbd}
\end{equation}
occurs with probability at least $1-\delta/n$. We claim that conditioning on \pref{eqn:BklUpperbd} for all $s\in[n]\setminus\{i,r\}$ and applying \pref{eqn:CkUpperbd} will give
\begin{equation}
|D_r|\lesssim nK^{3/2}\log\left(\frac{n}{\delta}\right) + \log^2\left(\frac{n}{\delta}\right) \label{eqn:DrEvent}
\end{equation}
with probability at least $1-\delta/n$. Indeed, let $\tilde{\bZ}$ be the principal submatrix of $\bZ$ indexed by $[n]\setminus\{i,r\}$, and let $w=(\bZ_{jr})_{j\in[n]\setminus\{i,r\}}$. Then $(B_{rs})_{s\in[n]\setminus\{i,r\}}=\tilde{\bZ} w$ as a column vector. By Theorem~\ref{thm:BvH} and Bernstein's inequality, with probability at least $1-\delta/n$,
$$\|\tilde{\bZ}\|\lesssim\sqrt{nK}+\sqrt{\log(n/\delta)}\lesssim\sqrt{nK},\qquad \|w\|_2^2\lesssim nK+\sqrt{nK\log(n/\delta)}+\log(n/\delta)\lesssim nK,$$
where we use $nK\geq\log(n/\delta)$. Thus $\sum_{s\in[n]\setminus\{i,r\}} B_{rs}^2\leq\|\tilde{\bZ}\|^2\|w\|_2^2\lesssim n^2K^2$. Combining this with \pref{eqn:BklUpperbd} and \pref{eqn:CkUpperbd} gives 
$$|D_r|\lesssim nK^{3/2}\log^{1/2}\left(\frac{n}{\delta}\right)+n^{1/2}K\log^{3/2}\left(\frac{n}{\delta}\right)+\log^2\left(\frac{n}{\delta}\right).$$
The assumption $nK\geq\log(n/\delta)$ implies $n^{1/2}K\log^{3/2}(n/\delta)\leq nK^{3/2}\log(n/\delta)$, proving \pref{eqn:DrEvent}.

Now for fixed values of $D_r$ for $r\in[n]\setminus\{i\}$, Bernstein's inequality implies that the event
\begin{equation}
|S_2'|=\left|\sum_{r\in[n]\setminus\{i\}}\bZ_{ir}D_r\right|\lesssim\sqrt{\sum_{r\in[n]\setminus\{i\}}KD_r^2\log\left(\frac{n}{\delta}\right)}+\max_r|D_r|\log\left(\frac{n}{\delta}\right) \label{eqn:S2ThirdUpperbd}
\end{equation}
occurs with probability at least $1-\delta$. (Notice that the collections $(D_r)_{r \in [n] \setminus \{i\}}$ and $(\bZ_{ir})_{r \in [n] \setminus \{i\}}$ are independent, so we can indeed fix the values of $D_r$.) Hence if we condition on \pref{eqn:DrEvent} for all $r\in[n]\setminus\{i\}$ and apply \pref{eqn:S2ThirdUpperbd} together with \pref{eqn:S2NullDecoupling}, then we find that the event
\begin{equation}
|S_2|\lesssim n^{3/2}K^2\log^{3/2}\left(\frac{n}{\delta}\right) + n^{1/2} K^{1/2} \log^{5/2}\left(\frac{n}{\delta}\right)  \label{eqn:Sprime2Event}
\end{equation}
occurs with probability at least $1-\delta$, using the assumption $n K \ge \log(n/\delta)$. Using a union bound over the above events, the probability of \pref{eqn:Sprime2Event} is at least $1-3\delta$. 

Hence combining our bounds for $S_1$ and $S_2$ together with the assumption $n \ge \log^3(n/\delta)$, we find that the event
$$\sigma_i^2\lesssim n^2K^2 + n^{1/2} K^{1/2} \log^{5/2}\left(\frac{n}{\delta}\right)$$
occurs with probability at least $1-4\delta$, completing the proof.
\end{proof}

\begin{proof}[Proof of \Cref{lemma:sigmastarTail}]
We will use that $\norm{\bZ_{ik}\bZ_{jk}}_\infty\leq1$ and $\Var \bZ_{ik}\bZ_{jk}\leq K^2$, which holds for all distinct $i,j,k \in [n]$. Using Bernstein's inequality, we have that for any $\delta \in (0,1)$, with probability at least $1-\delta/n^2$,
$$
|(\bZ^2)_{ij}| 
= \left|\sum_{k\in[n]\setminus\{i,j\}}\bZ_{ik}\bZ_{jk}\right| \le C \left( \sqrt{n K^2 \log(n/\delta)} + \log(n/\delta) \right)
$$
for an absolute constant $C>0$. Taking a union bound completes the proof.
\end{proof}


\section{Additional proofs for the semidefinite program}

This section is devoted to \Cref{lemma:sdp-eta-lb}, a key step in the proof of \Cref{thm:SDPmain}. We will refer repeatedly to our assumptions on $n$, $k$, $p$, $d$, and $\delta$:
\[\arraycolsep=5mm
\begin{array}{lll}
{\textrm{(A1)}\assumplabel{a1}{(A1)}}\hspace{2mm} np^{3/2}=\omega((d\log n)^2) \,, &
{\textrm{(A2)}\assumplabel{a2}{(A2)}}\hspace{2mm} kp^{3/4} \ge C n^{1/2}d \,,
\end{array}\]
where $C>0$ is a sufficiently large constant.

\begin{lemma}\label{lemma:sdp-eta-lb}
For all fixed $\epsilon\in(0,1)$, if $\eta$ is defined as in \eqref{eqn:semidefinite-definitions} and \ref{a1} and \ref{a2} hold, then for all large enough $n$ depending on $\epsilon$, we have $\eta>(1-\epsilon)\cdot p^3k^2/d^2$ with probability at least $1-n^{-10}/2$.
\end{lemma}
\begin{proof}
Let $\delta:=n^{-10}/2$. Recall that $\eta := \min_{1\leq i\leq k}(\bM\xi)_i-\max_{k<i\leq n} |(\bM\xi)_i|$. The proof resembles the spectral method in \Cref{sec:proofs} and analyzes the two cases $1\leq i\leq k$ and $k<i\leq n$. In both cases we have $(\bM\xi)_i=\sum_{j=1}^k \bM_{ij}$.

\smallskip
\noindent 
\textbf{Case 1:} $i \le k$. By \eqref{eq:M-noise-decomposition}, it holds that
\begin{align*}
\sum_{j=1}^k \bM_{ij} &\ge 
\underbrace{p^3 \sum_{j=1}^k (\bK^2 \circ \bK)_{ij}}_{T_1} - 
\underbrace{p^2 \bigg| \sum_{j=1}^k (\bK^2)_{ij} \bW_{ij} \bigg|}_{T_2} -
\underbrace{p \bigg| \sum_{j=1}^k (\bW^2)_{ij} \bK_{ij} \bigg|}_{T_3} -
\underbrace{\bigg| \sum_{j=1}^k (\bW^2)_{ij} \bW_{ij} \bigg|}_{T_4} \notag \\
&\quad -
\underbrace{p^2 \bigg| \sum_{j=1}^k (\bK^\# \bW + \bW \bK^\#)_{ij} \bK_{ij} \bigg|}_{T_5} -
\underbrace{p  \bigg| \sum_{j=1}^k (\bK^\# \bW + \bW \bK^\#)_{ij} \bW_{ij} \bigg|}_{T_6} . \notag
\end{align*}

\smallskip 
\noindent 
\textbf{Case 2:} $i > k$.
By \eqref{eq:M-noise-decomposition}, it holds that 
\begin{align*}
\bigg| \sum_{j=1}^k \bM_{ij} \bigg|
&\le \bigg| \underbrace{\sum_{j=1}^k (\bW^2)_{ij} \bW_{ij} \bigg|}_{T_7} + \underbrace{p\bigg| \sum_{j=1}^k (\bW \bK^\#)_{ij} \bW_{ij} \bigg|}_{T_8}
\end{align*}
since all the other terms vanish for $i>k$.

For all $1\leq i\leq k$ and $k<j\leq n$, define the event
$$E_{ij}:=\left\{(\bM\xi)_i-\abs{(\bM\xi)_j}\leq(1-\epsilon)\frac{p^3k^2}{d^2}\right\}\,.$$
Using the inclusion
$$E_{ij}\subseteq\left\{T_1\leq\left(1-\frac{\epsilon}{2}\right)\frac{p^3k^2}{d^2}\right\}\cup\bigcup_{l=2}^8\left\{T_l\geq\frac{\epsilon}{14}\frac{p^3k^2}{d^2}\right\}\,,$$
we calculate that
\begin{align*}
\bbP\left\{\eta\leq(1-\epsilon)\frac{p^3k^2}{d^2}\right\} &\leq \sum_{1\leq i\leq k}\sum_{k<j\leq n}\bbP\left\{E_{ij}\right\} \\
&\leq n^2\left(\bbP\left\{T_1\leq\left(1-\frac{\epsilon}{2}\right)\frac{p^3k^2}{d^2}\right\}+\sum_{i=2}^8\bbP\left\{T_i\geq\frac{\epsilon}{14}\frac{p^3k^2}{d^2}\right\}\right) \,.
\end{align*}
We will show in the second half of the proof that assumptions \ref{a1} and \ref{a2} imply the right-hand sides of \Cref{lemma:SDP-K2W,lem:SDP-W2K,lemma:SDP-W2W,lemma:SDP-KWWKK,lemma:SDP-KWWKW} are all $o(p^3k^2/d^2)$ after multiplying by the corresponding factor of $p$. For the moment, we use this fact to complete the proof of the lemma. We apply \Cref{lemma:sdp-signal-term} with $\lambda=1-\epsilon/2$ and \Cref{lemma:SDP-K2W,lem:SDP-W2K,lemma:SDP-W2W,lemma:SDP-KWWKK,lemma:SDP-KWWKW} with $\delta=n^{-12}/8$. Notice that $T_7$ is the same as $T_4$; additionally, $T_8$ appears as the object $S_2$ in the proof of \Cref{lemma:SDP-KWWKW}, so the bound from \Cref{lemma:SDP-KWWKW} also bounds $T_8$. We then obtain
$$\bbP\left\{\eta\leq(1-\epsilon)\frac{p^3k^2}{d^2}\right\}\leq n^2\left(e^{-ck/d}+\frac{7}{n^{12}}\right) \leq e^{-ck/2d}+\frac{7}{8n^{10}} \leq \frac{1}{n^{10}} \,,$$
where in the last step we used that $k/d\gg\log n$, which follows easily from \ref{a2} (this also implies that $k \geq d+\log(1/\delta)$, which we needed to apply these lemmas).

It remains to verify the right-hand sides of \Cref{lemma:SDP-K2W,lem:SDP-W2K,lemma:SDP-W2W,lemma:SDP-KWWKK,lemma:SDP-KWWKW} are $o(p^3k^2/d^2)$ after multiplying by the corresponding factor of $p$. We will repeatedly use the following easy consequences of \ref{a1} and \ref{a2}. Let $L:=\log(1/\delta)$, $L_k:=\log(k/\delta)$, and $L_n:=\log(n/\delta)$. Since $\delta=n^{-10}/2$, we have $L_n\lesssim\log n$. From \ref{a2}, we have $pk\geq Cn^{2/3}d^{4/3}k^{-1/3}$, hence
\begin{equation}\label{eqn:aux-p-lb}
\frac{d\log^{4/3}n}{pk}\lesssim \frac{d\log^{4/3}n}{n^{2/3}d^{4/3}k^{-1/3}}=\frac{k^{1/3}\log^{4/3}n}{n^{2/3}d^{1/3}}\leq\frac{\log^{4/3}n}{n^{1/3}}=o(1)\,.
\end{equation}
We verify the terms one at a time:
\begin{enumerate}[(i), leftmargin=1.5cm]
    \item[(L\ref{lemma:SDP-K2W})] Dividing by $p^3k^2/d^2$, the two constituent terms become $\sqrt{dL/pk}$ and $dL/pk$, respectively, which are $o(1)$ by $L\lesssim\log n$ and \eqref{eqn:aux-p-lb}.

    \item[(L\ref{lem:SDP-W2K})] For the first term, we calculate
    $$\frac{p^2\sqrt{nk/d}\,L_n}{p^3k^2/d^2}=\frac{d^{3/2}\sqrt{n}\,L_n}{pk^{3/2}}\lesssim \log n\cdot \Big(\frac{d}{nk}\Big)^{1/6}=o(1)\,,$$
    where the inequality used $L_n\lesssim \log n$ and \ref{a2} in the form $p\gtrsim(\sqrt{n}d/k)^{4/3}$. For the second term,
    $$\frac{p\sqrt{kp/d}\,L_n^2}{p^3k^2/d^2}=\bigg(\frac{dL_n^{4/3}}{pk}\bigg)^{3/2}=o(1)\,,$$
    where we used $L_n\lesssim\log n$ and \eqref{eqn:aux-p-lb}.

    \item[(L\ref{lemma:SDP-W2W})] For the first term, we calculate
    $$\frac{\sqrt{nk}p^{3/2}\sqrt{L}}{p^3k^2/d^2}=\frac{d^2\sqrt{nL}}{p^{3/2}k^{3/2}}=o(1)\,,$$
    where the $o(1)$ follows since $\frac{d^2\sqrt{nL}}{p^{3/2}k^{3/2}}=o(1)$ is equivalent to $k\gg (nL)^{1/3}d^{4/3}/p$, which in turn holds for the following reason: \ref{a2} asserts $k\gtrsim\sqrt{n}d/p^{3/4}$, and taking the ratio
    $$R:=\frac{(nL)^{1/3}d^{4/3}/p}{\sqrt{n}d/p^{3/4}}=\frac{(dL)^{1/3}}{n^{1/6}p^{1/4}}\,,$$
    we notice that $R=o(1)$ if and only if $np^{3/2}\gg(dL)^2$, which is precisely \ref{a1}. For the second term,
    $$\frac{p\sqrt{k}L_n^{3/2}}{p^3k^2/d^2}=\frac{d^2L_n^{3/2}}{p^2k^{3/2}}\lesssim \frac{k^{7/6}(\log n)^{3/2}}{n^{4/3}d^{2/3}}\leq\frac{\log^{3/2}n}{n^{1/6}d^{2/3}}=o(1)\,,$$
    where the first inequality uses $L_n\lesssim \log n$ and \ref{a2} in the form $p^2\gtrsim (\sqrt{n}d/k)^{8/3}$, and the second inequality uses $k\leq n$. For the third term, the condition $\frac{L_n^2}{p^3k^2/d^2}=o(1)$ is equivalent to $k\gg(dL_n)/p^{3/2}$, which holds for the following reason: again comparing with the condition $k\gtrsim\sqrt{n}d/p^{3/4}$ from \ref{a2}, the ratio
    $$\frac{dL_n/p^{3/2}}{\sqrt{n}d/p^{3/4}}=\frac{L_n}{\sqrt{n}p^{3/4}}$$
    is $o(1)$ if and only if $np^{3/2}\gg L_n^2$, which holds by \ref{a1}.

    \item[(L\ref{lemma:SDP-KWWKK})] Identical to (L\ref{lemma:SDP-K2W}).

    \item[(L\ref{lemma:SDP-KWWKW})] For the first term, we calculate
    $$\frac{p^2kL_k/\sqrt{d}}{p^3k^2/d^2}=\frac{d^{3/2}L_k}{pk}\lesssim \frac{d^{1/6}k^{1/3}\log n}{n^{2/3}}\leq \frac{d^{1/6}\log n}{n^{1/3}}=o(1)\,,$$
    where we used $L_k\lesssim\log n$ and \ref{a2} in the form $p\gtrsim(\sqrt{n}d/k)^{4/3}$. Finally, \ref{a2} implies $d\lesssim k/\sqrt{n}\leq \sqrt{n}$, hence $\frac{d^{1/6}\log n}{n^{1/3}}\lesssim \frac{\log n}{n^{1/4}}=o(1)$. For the second term,
    $$\frac{p\sqrt{pk/d}\,L_k^2}{p^3k^2/d^2}=\frac{d^{3/2}L_k^2}{p^{3/2}k^{3/2}}\lesssim \frac{\sqrt{k}(\log n)^2}{n\sqrt{d}}\leq \frac{(\log n)^2}{\sqrt{n}}=o(1)\,,$$
    where we used $L_k\lesssim\log n$ and then \ref{a2} in the form $p^{3/2}\gtrsim(\sqrt{n}d/k)^2=nd^2/k^2$.
\end{enumerate}
Having verified the right-hand sides of \Cref{lemma:SDP-K2W,lem:SDP-W2K,lemma:SDP-W2W,lemma:SDP-KWWKK,lemma:SDP-KWWKW} are of the correct order, the proof is complete.
\end{proof}

\begin{lemma}\label{lemma:sdp-signal-term}
Let $S:=\sum_{j=1}^k (\bK^2 \circ \bK)_{1j}$. For every fixed $\lambda\in(0,1)$ there exists a constant $c>0$ such that
$$\bbP\left\{S\leq\lambda\,\frac{k^2}{d^2}\right\} \leq e^{-ck/d}\,.$$
\end{lemma}
\begin{proof}
Throughout the proof we condition on $X_1$. We can rewrite
\begin{equation}\label{eqn:sdp-t1-S-expand}
S=\sum_{\substack{2\leq j,l\leq k\\j\neq l}}\inner{X_1,X_j}\inner{X_1,X_l}\inner{X_j,X_l}\,.
\end{equation}
Define the centered kernel $h:(\bbS^{d-1})^2\to\bbR$ by
\begin{equation}\label{eqn:sdp-t1-def-h}
h(x,y):=\inner{X_1,x}\inner{X_1,y}\inner{x,y}-\frac{1}{d}\inner{X_1,x}^2-\frac{1}{d}\inner{X_1,y}^2+\frac{1}{d^2}\,.
\end{equation}
For $Y\sim\Unif(\bbS^{d-1})$ and any fixed $x\in\bbS^{d-1}$, we have
\begin{align*}
& \bbE_Y[\inner{X_1,x}\inner{X_1,Y}\inner{x,Y}\,|\,X_1] = \inner{X_1,x}\cdot X_1^\top\bbE[YY^\top]x = \frac{1}{d}\inner{X_1,x}^2 \,, \\
& \bbE_Y[\inner{X_1,Y}^2\,|\,X_1] = \frac{1}{d}\inner{X_1,X_1}^2 = \frac{1}{d} \,,
\end{align*}
which implies
\begin{equation}\label{eqn:sdp-t1-canonical}
\bbE_Y[h(x,Y)\,|\,X_1]=0\,.
\end{equation}
By symmetry the same holds with $x$ and $y$ swapped.

Writing $W_j:=\inner{X_1,X_j}^2-1/d$, we directly expand \eqref{eqn:sdp-t1-S-expand} and \eqref{eqn:sdp-t1-def-h} to obtain
\begin{equation}\label{eqn:sdp-t1-S-hoeffding}
S = \mu + \frac{2(k-2)}{d}\sum_{j=2}^kW_j + U \,, \quad \mu:=\frac{(k-1)(k-2)}{d^2}\,, \quad U:=\sum_{\substack{2\leq j,l\leq k\\j\neq l}} h(X_j,X_l).
\end{equation}
In particular, $\bbE[S\,|\,X_1]=\mu$. If we set $t:=(1-\lambda)\mu$ then we find from \eqref{eqn:sdp-t1-S-hoeffding} that
\begin{equation}\label{eqn:sdp-t1-union}
\{S\leq \mu-t\}\subseteq \left\{U\leq -\frac{t}{2}\right\} \cup \bigg\{\frac{2(k-2)}{d}\sum_{j=2}^k W_j\leq -\frac{t}{2}\bigg\}.
\end{equation}
We will bound the two terms separately.

The random variables $W_2,\dots,W_k$ are independent given $X_1$, satisfy $\norm{W_j}_{\psi_1}\lesssim 1/d$, and are mean-zero. Hence by the subexponential Bernstein inequality (e.g.\ Theorem~2.8.1 in \cite{vershynin2018high}) conditional on $X_1$, there exists $c > 0$ such that
\begin{equation*}
\bbP\bigg\{\sum_{j=2}^kW_j\leq-s\bigg\} \leq \exp\left(-c\min\left\{\frac{s^2d^2}{k},\,sd\right\}\right)
\end{equation*}
for all $s>0$. Taking $s:=\frac{td}{4(k-2)}=\frac{(1-\lambda)(k-1)}{4d}$, we obtain
\begin{equation}\label{eqn:sdp-t1-linear-tail}
\bbP\bigg\{\frac{2(k-2)}{d}\sum_{j=2}^kW_j\leq -\frac{t}{2}\bigg\} \leq e^{-c(1-\lambda)^2k}\,.
\end{equation} 

We now move on to the U-statistic term. Let $\bX':=\{X'_j\}_{j\in[k]}$ and $\bX'':=\{X''_j\}_{j\in[k]}$ be independent copies of $\{X_j\}_{j\in[k]}$, and define the decoupled U-statistic
\begin{equation*}
U':=\sum_{\substack{2\leq j,l\leq k\\j\neq l}} h(X'_j,X''_l)\,.
\end{equation*}
By the decoupling inequality for order-2 $U$-statistics \citep[Theorem~1]{de1995decoupling}, there is an absolute constant $C'>0$ such that for all $x>0$,
\begin{equation}\label{eqn:sdp-t1-decouple}
\bbP\left\{|U|\geq x\right\}\leq C'\cdot \bbP\left\{C'\cdot |U'|\geq x\right\}.
\end{equation}
We now estimate tail probabilities for $U'$ using an exponential inequality for $U$-statistics due to \citeauthor{gine2000exponential}, which we apply conditionally on $X_1$. Equation~\eqref{eqn:sdp-t1-canonical} shows that, conditional on $X_1$, the kernel $h$ is canonical in the sense needed to apply \cite[Theorem~3.3]{gine2000exponential}. This result is written in terms of various norms and expectations of $h$, which we now estimate using the simpler forms of $A$, $B$, $C$, and $D$ appearing before \cite[Corollary~3.4]{gine2000exponential}.

\noindent\textbf{Bound for $A$.} It is easy to see that $A=\norm{h}_\infty \leq 1+2/d+1/d^2\leq 4$.

\noindent\textbf{Bounds for $B$ and $C$.} We will use that for any fixed $v,w\in\bbS^{d-1}$,
\begin{equation*}
\bbE[\inner{v,Y}^2\inner{w,Y}^2]=\frac{1+2\inner{v,w}^2}{d(d+2)}\,.
\end{equation*}
For any fixed $x\in\bbS^{d-1}$, using $(a+b+c+d)^2\leq 4(a^2+b^2+c^2+d^2)$, we find that there exists a universal constant $C>0$ such that
\begin{align}
\bbE_Y[h(x,Y)^2\,|\,X_1] & \leq 4\,\bbE_Y[\inner{X_1,x}^2\inner{X_1,Y}^2\inner{x,Y}^2] \nonumber\\
&\qquad + \frac{4}{d^2}\inner{X_1,x}^4 + \frac{4}{d^2}\bbE_Y[\inner{X_1,Y}^4] + \frac{4}{d^4}\nonumber\\
&\leq 4\inner{X_1,x}^2\,\frac{1+2\inner{X_1,x}^2}{d(d+2)} + O\left(\frac{1}{d^2}\right) \leq \frac{C}{d^2}.\label{eqn:sdp-t1-cond-h2}
\end{align}
The same estimate holds with $x$ and $Y$ interchanged, hence
\begin{equation*}
B^2 = (k-1)\big(\norm{\bbE_Yh(\cdot,Y)^2}_\infty + \norm{\bbE_X h(X,\cdot)^2}_\infty\big) \lesssim \frac{k}{d^2} \,,
\end{equation*}
where the implicit constant is universal. Similarly, taking expectation of \eqref{eqn:sdp-t1-cond-h2} over $x=X$ and using $\bbE[\inner{X_1,X}^2]=1/d$ gives $\bbE[h(X,Y)^2]\lesssim1/d^3$, which implies
\begin{equation*}
C^2=(k-1)^2\,\bbE[h(X,Y)^2]\lesssim\frac{k^2}{d^3} \,,
\end{equation*}
where the implicit constant is again universal.

\noindent\textbf{Bound for $D$.} From \cite[p.~20]{gine2000exponential}, we have $D=(k-1)\norm{h}_{L_2\to L_2}$ and $\norm{h}_{L_2\to L_2}\leq \norm{h}_{L_2}=\sqrt{\bbE[h(X,Y)^2]}$ by Cauchy--Schwarz. Thus since $\bbE[h(X,Y)^2] \lesssim 1/d^3$,
\begin{equation*}
D\leq (k-1)\norm{h}_{L_2} \lesssim \frac{k}{d^{3/2}}\,,
\end{equation*}
where the implicit constant is universal.

Using the above bounds on $A$, $B$, $C$, and $D$, we apply \cite[Theorem~3.3]{gine2000exponential} to $U'$, conditional on $X_1$, to obtain for all $x>0$,
\begin{equation*}
\bbP\left\{|U'|\geq x\,|\,X_1\right\} \leq L\exp\bigg(-\frac{1}{L}\min\bigg\{\frac{x^2}{C^2},\frac{x}{D},\frac{x^{2/3}}{B^{2/3}},\frac{x^{1/2}}{A^{1/2}}\bigg\}\bigg) \,,
\end{equation*}
where $L>0$ is an absolute constant. Applying our upper bounds for $A$, $B$, $C$, and $D$ (which do not depend on $X_1$), and then taking expectations over $X_1$ and using \eqref{eqn:sdp-t1-decouple} gives
\begin{equation*}
\bbP\left\{|U|\geq x\right\} \leq C\exp\bigg(-c\min\bigg\{\frac{x^2d^3}{k^2},\frac{xd^{3/2}}{k},\frac{x^{2/3}d^{2/3}}{k^{1/3}},x^{1/2}\bigg\}\bigg) \,.
\end{equation*}
Setting $x:=t/2=(1-\lambda)\mu/2$ and using $\mu\sim k^2/d^2$ yields
\begin{equation}\label{eqn:sdp-t1-U-plug}
\bbP\left\{U\leq -\frac{t}{2}\right\} \leq \bbP\left\{|U|\geq \frac{t}{2}\right\} \leq C\exp\left(-c\frac{k}{d}\right)\,,
\end{equation}
where we used that $\lambda$ is fixed and the last term in the minimum dominates. Combining \eqref{eqn:sdp-t1-union}, \eqref{eqn:sdp-t1-linear-tail}, and \eqref{eqn:sdp-t1-U-plug} completes the proof.
\end{proof}

\begin{lemma}\label{lemma:SDP-K2W}
There is an absolute constant $C>0$ such that the following holds. Assume $k\geq d+\log(1/\delta)$. For all $1\leq i\leq k$ and $\delta\in(0,1)$,
$$\bigg|\sum_{j=1}^k(\bK^2)_{ij}\bW_{ij}\bigg| \leq C\left(\sqrt{\frac{pk^3}{d^3}\log\left(\frac{1}{\delta}\right)}+\frac{k}{d}\log\left(\frac{1}{\delta}\right)\right)$$
with probability at least $1-\delta$.
\end{lemma}
\begin{proof}
Fix $1\leq i\leq k$. Let $S:=\sum_{j=1}^k(\bK^2)_{ij}\bW_{ij}$ and $B_j:=(\bK^2)_{ij}$ for all $j\in[k]\setminus\{i\}$. Throughout the proof, we condition on $X_i$. We have $\Var[\bW_{ij}\,|\,\bK]\leq2p$, so conditioned on $\bK$, Bernstein's inequality implies
$$|S|\lesssim\sqrt{p\Bigg(\sum_{j\in[k]\setminus\{i\}}B_j^2\Bigg)\log\left(\frac{1}{\delta}\right)}+\left(\max_{j\in[k]\setminus\{i\}}|B_j|\right)\log\left(\frac{1}{\delta}\right)$$
with probability at least $1-\delta/3$. In the remainder, we will show that the events
\begin{equation}
\sum_{j\in[k]\setminus\{i\}}B_j^2 \lesssim \frac{k^3}{d^3} + \frac{k^{5/2}}{d^3}\sqrt{\log\left(\frac{1}{\delta}\right)}\,,\qquad \max_{j\in[k]\setminus\{i\}}|B_j| \lesssim \frac{k}{d} \,,\label{eqn:VarAndMaxterms}
\end{equation}
occur with probability $1-2\delta/3$. It is easy to check that the statement of the lemma follows from the above bound on $|S|$, \eqref{eqn:VarAndMaxterms}, and the assumption $k\geq d+\log(1/\delta)$.

If we define the quantities
$$Y:=\sum_{l=1}^k\bK_{il}X_l\,,\qquad Z:=(\bK_{ij})_{j\in[k]\setminus\{i\}}\,,\qquad\bG:=\sum_{l\in[k]\setminus\{i\}}X_lX_l^\top\,,$$
then $B_j=\inner{X_j,Y}-\bK_{ij}$ and
\begin{equation}
\begin{aligned}
\sum_{j\in[k]\setminus\{i\}}B_j^2 &= Y^\top\bG Y-2\norm{Y}^2+\norm{Z}^2 \\
&\leq\norm{\bG}\norm{Y}^2+\norm{Z}^2\leq\norm{\bG}^2\norm{Z}^2+\norm{Z}^2\,.
\end{aligned}
\label{eqn:sumBijsq}
\end{equation}
To see the last inequality in \eqref{eqn:sumBijsq}, let $\bH\in\bbR^{(k-1)\times d}$ have rows $(X_l^\top)_{l\in[k]\setminus\{i\}}$, and notice
$$\norm{Y}^2=\sum_{l,m=1}^k\bK_{il}\bK_{im}\inner{X_l,X_m}=Z^\top(\bH\bH^\top)Z\leq\norm{\bH\bH^\top}\norm{Z}^2=\norm{\bG}\norm{Z}^2\,,$$
since $\norm{\bH\bH^\top}=\norm{\bH^\top\bH}=\norm{\bG}$. Using that $\bbE\bG=\frac{k-1}{d}I_d$, concentration of the sample covariance matrix (e.g., Theorem~5.7 in \cite{rigollet2023high}) implies the event
$$\norm{\bG}\lesssim\frac{k}{d}+\frac{k}{d}\left(\sqrt{\frac{d+\log(1/\delta)}{k}}+\frac{d+\log(1/\delta)}{k}\right)\asymp\frac{k}{d}$$
occurs with probability at least $1-\delta/6$, where we used $k\geq d+\log(1/\delta)$. Since $\bbE[\bK_{ij}^2]=1/d$ and $\norm{\bK_{ij}^2-1/d}_{\psi_1}\lesssim1/d$, the subexponential Bernstein inequality implies
\begin{equation}
\norm{Z}^2\lesssim\frac{k}{d}+\frac{1}{d}\sqrt{k\log\left(\frac{1}{\delta}\right)}+\frac{1}{d}\log\left(\frac{1}{\delta}\right) \label{eqn:Kij-SoS-Bern}
\end{equation}
with probability at least $1-\delta/6$. 
Continuing from \eqref{eqn:sumBijsq} and using our bounds on $\norm{\bG}$ and $\norm{Z}^2$, and dropping redundant terms (using the assumption $k\geq d+\log(1/\delta)$), we find that the event
\begin{align*}
\sum_{j\in[k]\setminus\{i\}}B_j^2 
&\lesssim \left(\frac{k}{d}\right)^2\left(\frac{k}{d}+\frac{1}{d}\sqrt{k\log\left(\frac{1}{\delta}\right)}\right) = \frac{k^3}{d^3} + \frac{k^{5/2}}{d^3} \sqrt{\log\left(\frac{1}{\delta}\right)}
\end{align*}
occurs with probability at least $1-\delta/3$, which yields the first bound in \eqref{eqn:VarAndMaxterms}.

We now prove the bound on $\max_j|B_j|$. We have $B_j=\sum_{l=1}^k\bK_{il}\bK_{jl}$. Conditional on $X_i$ and $X_j$, the summands $\bK_{il}\bK_{jl}$ have mean $\bK_{ij}/d$ and $\psi_1$-norm $\lesssim1/d$. Hence by the subexponential Bernstein inequality and a union bound, the event
$$\max_{j\in[k]\setminus\{i\}}|B_j|\lesssim\frac{k}{d}+\frac{1}{d}\left(\sqrt{k\log\left(\frac{k}{\delta}\right)}+\log\left(\frac{k}{\delta}\right)\right)\lesssim\frac{k}{d}$$
occurs with probability at least $1-\delta/3$, where we again used $k\geq d+\log(1/\delta)$.
\end{proof}

\begin{lemma}\label{lem:SDP-W2K}
There is an absolute constant $C>0$ such that if \ref{a2} holds, then we have the following. Assume $k\geq d+\log(1/\delta)$. For all $1\leq i\leq k$ and $\delta\in(0,1)$,
$$\bigg|\sum_{j=1}^k(\bW^2)_{ij}\bK_{ij}\bigg| \leq C\left(p\sqrt{\frac{nk}{d}}\log\left(\frac{n}{\delta}\right) + \sqrt{\frac{kp}{d}}\log^2\left(\frac{n}{\delta}\right)\right)$$
with probability at least $1-\delta$.
\end{lemma}
\begin{proof}
Fix $1\leq i\leq k$. Let $\bW'$ be a conditionally independent copy of $\bW$ given $\bK$. Define the sums
\begin{equation}
S:=\sum_{j=1}^k\sum_{l=1}^n\bW_{il}\bW_{jl}\bK_{ij}\,,\qquad S':=\sum_{j=1}^k\sum_{l=1}^n\bW_{il}\bW'_{jl}\bK_{ij}\,. \label{eqn:lem:SDP-W2K}
\end{equation}
For all $l\in[n]\setminus\{i\}$ let $B_l:=\sum_{j=1}^k\bW'_{jl}\bK_{ij}$, so that $S'=\sum_{l=1}^n\bW_{il}B_l$. We claim that there exists a universal $c$ such that, for all $t > 0$, we have the decoupling inequality
\begin{equation}
\label{eqn:good-decoupling}
    \bbP(\abs{S} \geq t) \leq c\bbP(c\abs{S'} \geq t).
\end{equation}
This will follow from Theorem~3.4.1 in \cite{de2012decoupling}, applied conditionally on $\bK$, if we can write $S$ and $S'$ as sums of functions of independent random variables. As in the proof of Lemma \ref{lemma:decoupling}, when $\{i,j\} \in \binom{[n]}{2}$, we write $\bW_{\{i,j\}}$ for the common value $\bW_{ij} = \bW_{ji}$. These are conditionally independent given $\bK$, and take values in $\Omega = [-1,1]$. Let $J_n \subset \binom{[n]}{2}^2$ be the set of ordered pairs of \emph{distinct} elements. Say that $(\{a,b\},\{c,d\}) \in J_2$ is ``good'' if $\{a,b\} \cap \{c,d\} \neq \emptyset$ and $i \in \{a,b\} \setminus \{c,d\}$. Such pairs must actually be of the form $(\{a,b\},\{c,d\}) = (\{i,\ell\},\{j,\ell\})$ for some $j, \ell \in [n] \setminus \{i\}$. Write $J_n^{\textup{good}}$ for the set of all such pairs, and $\eta : J_n^{\textup{good}} \to [n] \setminus \{i\}$ for the function which ``extracts $j$'' in the sense of returning the element which is neither $i$ nor shared. This is all well-defined, and it allows us to define $h_{\{a,b\},\{c,d\}} : \Omega^2 \to \R$ by 
\[
    h_{\{a,b\},\{c,d\}}(x,y) = xy\bsone\{(\{a,b\},\{c,d\}) \text{ is good}\} (\bK^\#)_{i,\eta(\{a,b\},\{c,d\})}.
\]
With this definition, one can check that 
\begin{align*}
    S &= \sum_{j,\ell=1}^n \bW_{i\ell} \bW_{j\ell} (\bK^\#)_{ij} = \sum_{(\{a,b\},\{c,d\}) \in J_n} h_{\{a,b\},\{c,d\}}(\bW_{ab},\bW_{cd}) \\
    S' &= \sum_{j,\ell=1}^n \bW_{i\ell} \bW_{j\ell} (\bK^\#)_{ij} = \sum_{(\{a,b\},\{c,d\}) \in J_n} h_{\{a,b\},\{c,d\}}(\bW_{ab},\bW'_{cd}),
\end{align*}
so that Theorem~3.4.1 of \cite{de2012decoupling} indeed implies \eqref{eqn:good-decoupling}. 

Conditional on $\bK$ and $\bW'$, the random variables $\{B_l\bW_{il}\}_{l\in[n]\setminus\{i\}}$ are independent, centered, bounded in absolute value by $|B_l|$, and have variance at most $2pB_l^2$. Hence Bernstein's inequality implies
\begin{equation}
|S'| \lesssim \sqrt{p\Bigg(\sum_{l\in[n]\setminus\{i\}}B_l^2\Bigg)\log\Big(\frac{1}{\delta}\Big)} + \Big(\max_{l\in[n]\setminus\{i\}}|B_l|\Big)\log\Big(\frac{1}{\delta}\Big) \label{eqn:bern-Sprime-W2K}
\end{equation}
with conditional probability at least $1-\delta/(3c)$. 

We first bound the variance term. For each $l\in[n]\setminus\{i\}$, let us define $Y_l:=\sum_{j=1}^k\bW'_{jl}X_j$ so that $B_l=X_i^\top Y_l$. We then calculate that
$$\sum_{l\in[n]\setminus\{i\}}B_l^2\leq X_i^\top\bX^\top\bW'_{1:k,:}(\bW'_{1:k,:})^\top\bX X_i\leq\norm{\bX}^2\norm{\bW'}^2\,.$$
By concentration of the sample covariance matrix, with probability at least $1-\delta/(6c)$,
$$\norm{\bX}^2=\norm{\bX^\top\bX}\lesssim\frac{k}{d}+\frac{1}{d}\log\Big(\frac{1}{\delta}\Big)\lesssim\frac{k}{d}$$
where we used $k\geq d+\log(1/\delta)$.

Conditional on $\bK$, the entries of $\bW'$ are independent, centered, bounded by 1, and have variance at most $2p$. Applying \Cref{lemma:W-opnorm} to $\bW'$ with $\delta/(6c)$ in place of $\delta$, we obtain
$$\norm{\bW'}\lesssim\sqrt{np}+\sqrt{\log\left(\frac{n}{\delta}\right)}$$
with conditional probability at least $1-\delta/(6c)$ given $\bK$. Combining these bounds, we obtain
\begin{equation}\label{eqn:Bl-sos}
\sum_{l\in[n]\setminus\{i\}}B_l^2 \lesssim \norm{\bX}^2\norm{\bW'}^2 \lesssim \frac{k}{d}\cdot\left(np+\log\left(\frac{n}{\delta}\right)\right)
\end{equation}
with probability at least $1-\delta/(3c)$.

To control the second term in \eqref{eqn:bern-Sprime-W2K}, if we condition on $\bK$, Bernstein's inequality and a union bound yield
$$\max_{l\in[n]\setminus\{i\}}|B_l| \lesssim \sqrt{p\Bigg(\sum_{j\in[k]\setminus\{i\}}\bK_{ij}^2\Bigg)\log\left(\frac{n}{\delta}\right)}+\log\left(\frac{n}{\delta}\right)$$
with probability at least $1-\delta/(6c)$. The bound from \eqref{eqn:Kij-SoS-Bern} applies directly to $\sum_{j\in[k]\setminus\{i\}}\bK_{ij}^2$ with probability at least $1-\delta/(6c)$, so we obtain
$$\max_{l\in[n]\setminus\{i\}}|B_l| \lesssim \sqrt{\frac{pk}{d}}\log\left(\frac{n}{\delta}\right)+\log\left(\frac{n}{\delta}\right)$$
with probability at least $1-\delta/(3c)$, where we used $k\geq d+\log(1/\delta)$. By combining the bounds for the variance and max terms, using the decoupling inequality, and $k\geq d+\log(1/\delta)$, we obtain
\begin{equation}\label{eqn:SDP-W2K-final-bd}
p\sqrt{\frac{nk}{d}} \log\left(\frac{n}{\delta}\right) + \sqrt{\frac{pk}{d}} \log^{3/2}\left(\frac{1}{\delta}\right) + \sqrt{\frac{pk}{d}} \log^{2}\left(\frac{n}{\delta}\right) + \log^{2}\left(\frac{n}{\delta}\right).
\end{equation}
The second term in \eqref{eqn:SDP-W2K-final-bd} is clearly at most the third term. The fourth term is at most the third term since $\sqrt{pk/d}\geq1$, which in turn follows from \ref{a2}. Indeed, \ref{a2} and $k\leq n$ imply $p\gtrsim n^{-2/3}$, hence we have $pk/d\gtrsim \sqrt{n}p^{1/4}\gtrsim n^{1/3}\geq1$. This completes the proof of the lemma.
\end{proof}

\begin{lemma}\label{lemma:SDP-W2W}
There is an absolute constant $C>0$ such that we have the following. Let $1\leq i\leq k$ and $\delta\in(0,1)$. If \ref{a1}, \ref{a2}, and $\delta\geq n^{-C}$ hold, then
$$\bigg|\sum_{j=1}^k(\bW^2)_{ij}\bW_{ij}\bigg| \leq C\left(\sqrt{nk}p^{3/2}\sqrt{\log\left(\frac{1}{\delta}\right)}+p\sqrt{k}\log^{3/2}\!\left(\frac{n}{\delta}\right)+\log^2\!\left(\frac{n}{\delta}\right)\right)$$
with probability at least $1-\delta$.
\end{lemma}
\begin{proof}
Fix $1\leq i\leq k$ and let
$$S:=\sum_{j=1}^k\sum_{l=1}^n\bW_{il}\bW_{jl}\bW_{ij}\,,\qquad S':=\sum_{j=1}^k\sum_{l=1}^n\bW_{il}\bW'_{lj}\bW''_{ij}\,,$$
where $\bW',\bW''\in\bbR^{n\times k}$ are conditionally independent copies of $\bW$ given $\bK$. We claim that, for a universal constant $c>0$ and all $t>0$,
\[
    \bbP\{\abs{S}\geq t\}\leq c\,\bbP\{c\abs{S'}\geq t\}.
\]
Conditional on $\bK$, the random variables $\{\bW_{\{a,b\}}\}_{\{a,b\}\in\binom{[n]}{2}}$ are independent. Let $J_n\subset\binom{[n]}{2}^3$ be the set of ordered triples of distinct unordered pairs. For $(\{a,b\},\{c,d\},\{e,f\})\in J_n$, define
\begin{align*}
    h_{\{a,b\},\{c,d\},\{e,f\}}(x,y,z) 
    &= xyz\cdot \bsone\{\{a,b\},\{c,d\},\{e,f\}\text{ form a triangle}\} \\
    &\qquad \cdot
    \bsone\{i\in\{a,b\}\cap\{e,f\}\}
    \cdot \bsone\{\{e,f\}=\{i,j\}\text{ for some }j\in[k]\}.
\end{align*}
The nonzero terms are exactly those with $(\{a,b\},\{c,d\},\{e,f\})=(\{i,l\},\{l,j\},\{i,j\})$ for some $j\in[k]$ and $l\in[n]\setminus\{i,j\}$. Therefore, using the convention that the diagonal of $\bW$ vanishes,
\begin{align*}
    \sum_{(\{a,b\},\{c,d\},\{e,f\})\in J_n} h_{\{a,b\},\{c,d\},\{e,f\}}(\bW_{ab},\bW_{cd},\bW_{ef})
    &= \sum_{j=1}^k\sum_{l=1}^n \bW_{il}\bW_{lj}\bW_{ij}=S,
\end{align*}
and similarly
\begin{align*}
    \sum_{(\{a,b\},\{c,d\},\{e,f\})\in J_n} h_{\{a,b\},\{c,d\},\{e,f\}}(\bW_{ab},\bW'_{cd},\bW''_{ef})
    &= \sum_{j=1}^k\sum_{l=1}^n \bW_{il}\bW'_{lj}\bW''_{ij}=S'.
\end{align*}
Theorem~3.4.1 in \cite{de2012decoupling}, applied conditionally on $\bK$ to the independent variables indexed by $\binom{[n]}{2}$, gives the desired conditional decoupling inequality. Taking expectations over $\bK$ gives the claimed unconditional version.

Let $B_l:=\sum_{j=1}^k\bW'_{lj}\bW''_{ij}$ for $l\in[n]$, so that $S'=\sum_{l\in[n]}B_l\bW_{il}$. Conditional on $\bK$, $\bW'$, and $\bW''$, the variables $\{B_l\bW_{il}\}_{l\in[n]\setminus\{i\}}$ are independent, centered, bounded in absolute value by $|B_l|$, and have variance at most $2pB_l^2$. From Bernstein's inequality,
\begin{equation}
|S'| \lesssim \sqrt{p\Bigg(\sum_{l\in[n]\setminus\{i\}}B_l^2\Bigg)\log\left(\frac{1}{\delta}\right)} + \Big(\max_{l\in[n]\setminus\{i\}}|B_l|\Big)\log\left(\frac{1}{\delta}\right) \label{eqn:bern-Sprime}
\end{equation}
with probability at least $1-\delta/(3c)$.

Let us write $B:=(B_j)_{j\in[n]\setminus\{i\}}$ and $V:=\norm{B}^2$. Define the random vector $Y:=(\bW''_{ij})_{j\in[k]}$ with $Y_i:=0$. We then have
$$V=\norm{\bW'Y}^2-B_i^2\leq\norm{\bW'}^2\norm{Y}^2\,.$$
From \Cref{lemma:W-opnorm} below, we have the bound $\norm{\bW'}\lesssim\sqrt{np}$ with conditional probability at least $1-\delta/(6c)$ given $\bK$, where the $\sqrt{\log(n/\delta)}$ term is $O(\sqrt{np})$ due to \ref{a1}. Conditional on $\bK$, Bernstein's inequality implies
\[
    \norm{Y}^2\lesssim kp+\sqrt{kp\log\left(\frac{1}{\delta}\right)}+\log\left(\frac{1}{\delta}\right) \lesssim kp + \log\left(\frac{1}{\delta}\right)
\]
with probability at least $1-\delta/(6c)$. The fourth bullet point in the proof of Proposition~\ref{prop:main-spectral} shows that $\log(n) = o(kp)$ when \ref{a2} holds; thus, since $\delta \geq n^{-C}$, we have that $\|Y\|^2 \lesssim kp$ with probability at least $1-\delta/(6c)$. It then follows that
\[
    V \lesssim nkp^2
\]
with probability at least $1-\delta/(3c)$.

We now bound the second term in \eqref{eqn:bern-Sprime}. Conditional on $\bK$, the summands $\bW'_{lj}\bW''_{ij}$ in $B_l$ are centered, bounded in absolute value by 1, and have variance at most $4p^2$, hence Bernstein's inequality and a union bound imply
$$\max_{l\in[n]\setminus\{i\}}|B_l|\lesssim\sqrt{kp^2\log\left(\frac{n}{\delta}\right)}+\log\left(\frac{n}{\delta}\right)$$
with probability at least $1-\delta/(3c)$. Combining the above bounds and taking a union bound, we have
$$|S'| \lesssim \sqrt{nk}p^{3/2}\sqrt{\log\left(\frac{1}{\delta}\right)}+\sqrt{k}p\log^{3/2}\left(\frac{n}{\delta}\right)+\log^2\left(\frac{n}{\delta}\right)$$
with probability at least $1-\delta/c$. The lemma now follows by applying the decoupling inequality cited above to $S$ and $S'$.
\end{proof}

\begin{lemma}\label{lemma:W-opnorm}
There is an absolute constant $C>0$ such that the following holds. For all $\delta\in(0,1)$,
$$\norm{\bW}\leq C\left(\sqrt{np}+\sqrt{\log\left(\frac{n}{\delta}\right)}\right)$$
with conditional probability at least $1-\delta$ given $\bK$.
\end{lemma}
\begin{proof}
Condition on $\bK$ throughout the proof. If we define the quantities
$$\sigma:=\max_{i=1,\dots,n}\sqrt{\sum_{j=1}^n\bbE[\bW_{ij}^2\,|\,\bK]}\,,\hspace{10mm}\sigma_\ast:=\max_{i,j=1,\dots,n}\norm{\bW_{ij}}_\infty\,,$$
then $\sigma_\ast\leq 1$ and $\sigma^2\leq2np$, since $\bbE[\bW_{ij}^2\,|\,\bK]\leq2p$.
The hypotheses of \Cref{thm:BvH} are satisfied for $\bW$ since it is symmetric with independent centered entries above the diagonal. Hence applying \Cref{thm:BvH} and taking $t\asymp \sqrt{\log(n/\delta)}$ implies the event
$$\norm{\bW}\lesssim \sigma+t\lesssim \sqrt{np}+\sqrt{\log\left(\frac{n}{\delta}\right)}$$
occurs with probability at least $1-\delta$.
\end{proof}

\begin{lemma}\label{lemma:SDP-KWWKK}
There is an absolute constant $C>0$ such that the following holds. Assume $k\geq d+\log(1/\delta)$. For all $1\leq i\leq k$ and $\delta\in(0,1)$,
$$\bigg|\sum_{j=1}^k(\bK^\#\bW+\bW\bK^\#)_{ij}\bK_{ij}\bigg|\leq C\left(\sqrt{\frac{pk^3}{d^3}\log\left(\frac{1}{\delta}\right)}+\frac{k}{d}\log\left(\frac{1}{\delta}\right)\right)$$
with probability at least $1-\delta$.
\end{lemma}
\begin{proof}
Fix $1\leq i\leq k$. Define the sums
$$S:=\sum_{j=1}^k(\bK^\#\bW+\bW\bK^\#)_{ij}\bK_{ij}\,,\quad S_1:=\sum_{1\leq j<l\leq k}\bK_{il}\bW_{jl}\bK_{ij}\,,\quad S_2:=\sum_{j,l=1}^k\bW_{il}\bK_{jl}\bK_{ij}\,,$$
so that $S=2S_1+S_2$. Notice $|S_2|$ is precisely the quantity analyzed in \Cref{lemma:SDP-K2W} (applied with $\delta/2$), so it remains to bound $|S_1|$. Conditional on $\bK$, the summands in $S_1$ are independent, mean-zero, and variance at most $2p\bK_{il}^2 \bK_{ij}^2$. Hence by applying Bernstein's inequality given $\bK$, we have that with probability at least $1-\delta/2$,
\begin{align*}
|S_1| &\lesssim \sqrt{p\bigg(\sum_{1\leq j<l\leq k}\bK_{ij}^2\bK_{il}^2\bigg)\log\left(\frac{1}{\delta}\right)}+\log\left(\frac{1}{\delta}\right) \\
&\lesssim \left(\frac{k}{d}+\frac{1}{d}\sqrt{k\log\left(\frac{1}{\delta}\right)}+\frac{1}{d}\log\left(\frac{1}{\delta}\right)\right)\sqrt{p\log\left(\frac{1}{\delta}\right)}+\log\left(\frac{1}{\delta}\right) \\
&\lesssim \frac{k\sqrt{p}}{d}\sqrt{\log\left(\frac{1}{\delta}\right)} + \log\left(\frac{1}{\delta}\right) \,,
\end{align*}
where in the second inequality we used $\sum_{j<l}\bK_{ij}^2\bK_{il}^2\leq(\sum_{j}\bK_{ij}^2)^2$ and applied \eqref{eqn:Kij-SoS-Bern} to the right-hand side, and in the third inequality we used $k\geq d+\log(1/\delta)$. Combining the bounds for $|S_1|$ and $|S_2|$, and dropping redundant terms, implies the statement of the lemma.
\end{proof}

\begin{lemma}\label{lemma:SDP-KWWKW}
There is an absolute constant $C>0$ such that the following holds. Assume \ref{a2} and $k\geq d+\log(1/\delta)$. For all $1\leq i\leq k$ and $\delta\in(0,1)$,
$$\bigg|\sum_{j=1}^k(\bK^\#\bW+\bW\bK^\#)_{ij}\bW_{ij}\bigg|\leq C\left(\frac{pk}{\sqrt{d}}\log\!\left(\frac{k}{\delta}\right)+\sqrt{\frac{pk}{d}}\log^2\!\left(\frac{k}{\delta}\right)\right)$$
with probability at least $1-\delta$.
\end{lemma}
\begin{proof}
Fix $1\leq i\leq k$. Define the sums
$$S:=\sum_{j=1}^k(\bK^\#\bW+\bW\bK^\#)_{ij}\bW_{ij}\,,\quad S_1:=\sum_{j,l=1}^k\bK_{il}\bW_{jl}\bW_{ij}\,,\quad S_2:=\sum_{j,l=1}^k\bW_{il}\bK_{jl}\bW_{ij}\,,$$
so that $S=S_1+S_2$. First note that $S_1$ is precisely $S$ as defined in \eqref{eqn:lem:SDP-W2K}, except with the sum over $l\in[k]$ instead of $j\in[n]$. Hence it follows from \Cref{lem:SDP-W2K} (applied with $\delta/2$ and $n=k$) that
$$|S_1| \lesssim \frac{pk}{\sqrt{d}}\log\left(\frac{k}{\delta}\right) + \sqrt{\frac{pk}{d}}\log^2\left(\frac{k}{\delta}\right)$$
with probability at least $1-\delta/2$.

We will now bound $|S_2|$. Let $c>0$ be the constant from the decoupling inequality for order-2 kernels (Theorem~3.4.1 in \cite{de2012decoupling}). Let $\bW'$ be an independent copy of $\bW$ (that is, $\bW$ and $\bW'$ are conditionally i.i.d.\ given $\bK$). Let $B_l:=\sum_{j=1}^k\bW'_{ij}\bK_{jl}$ for all $l\in[k]\setminus\{i\}$, and define $S_2':=\sum_{l=1}^kB_l\bW_{il}$. Conditional on $\bK$ and $\bW'$, Bernstein's inequality implies
$$|S_2'|\lesssim\sqrt{p\Bigg(\sum_{l\in[k]\setminus\{i\}}B_l^2\Bigg)\log\left(\frac{1}{\delta}\right)}+\Big(\max_{l\in[k]\setminus\{i\}}|B_l|\Big)\log\left(\frac{1}{\delta}\right)$$
with probability at least $1-\delta/(6c)$. 

Let $Y:=\sum_{j=1}^k\bW'_{ij}X_j$, and let $\bG\in\bbR^{d\times d}$ and $\bH\in\bbR^{(k-1)\times d}$ be defined as in \Cref{lemma:SDP-K2W}. Also let $Z:=(\bW'_{ij})_{j\in[k]\setminus\{i\}}$. Notice that $B_l=\inner{Y,X_l}-\bW'_{il}$, which implies
\begin{align*}
\sum_{l\in[k]\setminus\{i\}}B_l^2 &\leq 2Y^\top\left(\sum_{l\in[k]\setminus\{i\}}X_lX_l^\top\right)Y+2\sum_{l\in[k]\setminus\{i\}}(\bW'_{il})^2 \\[5pt]
&= 2Y^\top\bG Y+2\norm{Z}^2 \leq 2\norm{\bG}\norm{Y}^2+2\norm{Z}^2\,.
\end{align*}
From \Cref{lemma:SDP-K2W} we have $\norm{\bG}\lesssim\frac{k}{d}$ with probability at least $1-\delta/(18c)$. Conditional on $\bX$, the random vectors $V_j:=\bW'_{ij}X_j$ for $j\in[k]\setminus\{i\}$ are independent, mean-zero, and satisfy $\norm{V_j}\leq1$ and $\bbE[\norm{V_j}^2\,|\,\bX]\leq2p$. Thus from the vector Bernstein inequality (e.g.\ \cite[Theorem~1.6]{tropp2012user}),
$$\norm{Y}\lesssim\sqrt{kp\log\left(\frac{d}{\delta}\right)}+\log\left(\frac{d}{\delta}\right)$$
with probability at least $1-\delta/(18c)$. The scalar Bernstein inequality implies
$$\norm{Z}^2\lesssim kp + \sqrt{kp\log\left(\frac{1}{\delta}\right)}+\log\left(\frac{1}{\delta}\right)$$
with probability at least $1-\delta/(18c)$. Combining the above three bounds, the event 
$$\sum_{l\in[k]\setminus\{i\}}B_l^2 \lesssim \frac{k^2p}{d}\log\left(\frac{d}{\delta}\right)+\frac{k}{d}\log^2\left(\frac{d}{\delta}\right)$$
occurs with probability at least $1-\delta/(6c)$.

For the maximum term, conditional on $\bK$, Bernstein's inequality gives
\[
    |B_l| \lesssim \sqrt{p \left(\sum_{j\in[k]\setminus\{i\}}\bK_{jl}^2\right)\log\left(\frac{k}{\delta}\right)}+\log\left(\frac{k}{\delta}\right)
\] 
for all $l\in[k]\setminus\{i\}$ with conditional probability at least $1-\delta/(12c)$. Further, applying \eqref{eqn:Kij-SoS-Bern} with failure probability $\delta/(12ck)$ and taking a union bound over $l\in[k]\setminus\{i\}$ gives
\[
\max_{l\in[k]\setminus\{i\}}\sum_{j\in[k]\setminus\{i\}}\bK_{jl}^2
\lesssim \frac{k}{d}+\frac{1}{d}\sqrt{k\log\left(\frac{k}{\delta}\right)}+\frac{1}{d}\log\left(\frac{k}{\delta}\right)
\lesssim \frac{k}{d}
\]
with probability at least $1-\delta/(12c)$, where the last inequality uses $k\geq d+\log(1/\delta)$. Combining the last two displays and absorbing constants, we obtain
\[
\max_{l\in[k]\setminus\{i\}}|B_l| \lesssim \sqrt{\frac{pk}{d}\log\left(\frac{k}{\delta}\right)}+\log\left(\frac{k}{\delta}\right)
\lesssim \sqrt{\frac{pk}{d}}\log\left(\frac{k}{\delta}\right)+\log\left(\frac{k}{\delta}\right)
\]
with probability at least $1-\delta/(6c)$. 

Combining the above bounds and using $k \geq d+\log(1/\delta)$ implies
$$|S_2'|\lesssim\frac{kp}{\sqrt{d}}\log\left(\frac{d}{\delta}\right)+\sqrt{\frac{pk}{d}}\log^2\left(\frac{k}{\delta}\right)+\log^2\left(\frac{k}{\delta}\right)\,,$$
with probability at least $1-\delta/(2c)$. By the decoupling inequality, the same bound holds for $|S_2|$ with probability at least $1-\delta/2$. It remains to absorb the third term into the second on the right-hand side using the same argument as at the end of the proof of \Cref{lem:SDP-W2K} together with \ref{a2}.
\end{proof}


\section{Computational lower bound}
\label{sec:computational-lower bound}
We prove \Cref{thm:CompLwrBd} in this section following \cite{schramm2022computational}. 
Before stating the proof, we need several definitions. Let $\bB\sim G(n,p)$ coupled to $\bA$ such that for all distinct pairs $1\leq i,j\leq n$ not both in $\cS$, we have $\bA_{ij}=\bB_{ij}$, and the collections $\{\bA\}_{i,j\in[n]}$ and $\{\bB\}_{i,j\in\cS}$ are independent. For a subset $\alpha\subseteq\binom{[n]}{2}$ (which we identify with the binary vector with entries $\bsone\{e\in\alpha\}$, $e\in\binom{[n]}{2}$) and an $n\times n$ matrix $\bY$, define
$$\bY^\alpha:=\prod_{ij\in\alpha}\bY_{ij}\,.$$
We will also view $\alpha$ as a graph whose vertex set, denoted $V(\alpha)$, is induced by its edges. With these definitions, the following calculation is the correct modification of the top of p.~17 in the technical appendix of \cite{schramm2022computational}:
\begin{align*}
\bbE[\bar{\bA}^\alpha\,|\,\bB] &= \bbE_\cS\Big[\bar{\bB}^{\alpha\setminus\binom{\cS}{2}}\cdot\bbE\Big[\bar{\bA}^{\alpha\cap\binom{\cS}{2}}\,\Big|\,\cS\Big]\,\Big|\,\bB\Big] \\
&= \sum_{\beta\subseteq\alpha}\bbE\Big[\bsone\{\alpha\setminus\tbinom{\cS}{2}=\beta\}\cdot\bar{\bA}^{\alpha\cap\binom{\cS}{2}}\Big] \cdot \bar{\bB}^\beta = \sum_{\beta\subseteq\alpha}M_{\alpha\beta}\bar{\bB}^\beta \,,
\end{align*}
where
$$\textstyle M_{\alpha\beta}:=\bbE\left[\bsone\big\{\alpha\setminus\binom{\cS}{2}=\beta\big\}\cdot\bar{\bA}^{\alpha\cap\binom{\cS}{2}}\right]\,.$$
The following lemma now follows from the proof of Theorem 2.2 of \cite{schramm2022computational} (see also p.~17 of the technical appendix).

\begin{lemma}[\cite{schramm2022computational}]
Let $\bA$ and $\bB$ be as defined above. Suppose $M_{\alpha\alpha}\neq0$ for all $\alpha\subseteq\binom{[n]}{2}$ with $|\alpha|\leq D$. Then
$$\CorrleD^2\leq\sum_{\alpha\subseteq\binom{[n]}{2},\,|\alpha|\leq D}w_\alpha^2\,,$$
where $w_\alpha$ are defined recursively as
$$w_\emptyset=\bbE[\theta]\,,\quad w_\alpha:=\frac{1}{M_{\alpha\alpha}}\left(\bbE[\bar{\bA}^\alpha\cdot\theta]-\sum_{\beta\subsetneq\alpha}M_{\alpha\beta}w_\beta\right)\,,\quad {\textstyle\alpha\subseteq\binom{[n]}{2}}\,,\quad |\alpha|\geq1\,,$$
where $\beta\subsetneq\alpha$ means $\beta$ is a proper subgraph of $\alpha$.
\end{lemma}

\begin{proof}[Proof of \Cref{thm:CompLwrBd}]
The proof of $\CorrleD^2\leq(1+o(1))r^2$ is nearly the same as that of Lemma H.2({\em i}) in the technical appendix of \cite{schramm2022computational}, so we only describe the parts that are different. We only need to confirm that Lemma H.4 of \cite{schramm2022computational} still holds in our setting, for which it suffices to prove the same bounds on $M_{\alpha\alpha}$, $|M_{\alpha\beta}|$, and $|\bbE[\bar{\bA}^\alpha\cdot\theta]|$. Note that Lemma H.3 in their article holds true unchanged in our setting.

For all $\alpha\subseteq\binom{[n]}{2}$ we have
\begin{align*}
M_{\alpha\alpha} &= \bbE\left[\bsone\big\{\alpha\setminus{\textstyle\binom{\cS}{2}}=\alpha\big\}\cdot\bar{\bA}^{\alpha\cap\binom{\cS}{2}}\right] \\
&= \bbP\big\{\alpha\setminus{\textstyle\binom{\cS}{2}}=\alpha\big\}\geq\bbP\{V(\alpha)\cap\cS=\emptyset\}=(1-r)^{|V(\alpha)|}\,.\label{eqn:Malphaalpha}
\end{align*}
For fixed $\beta\subseteq\alpha$, since $\alpha\setminus\binom{\cS}{2}=\beta$ implies $V(\alpha\setminus\beta)\subseteq\cS$, we calculate
\begin{align*}
|M_{\alpha\beta}| &= \left|\bbE_\cS\left[\bsone\{\alpha\setminus\tbinom{\cS}{2}=\beta\}\cdot\bbE\left[\bar{\bA}^{\alpha\cap\binom{\cS}{2}}\,\Big|\,\cS\right]\right]\right| \\
&= \bbP\big\{\alpha\setminus\tbinom{\cS}{2}=\beta\big\}\cdot\abs{\bbE\left[\bar{\bA}^\alpha\,|\,V(\alpha)\subseteq\cS\right]} \\
&= \bbP\big\{\alpha\setminus\tbinom{\cS}{2}=\beta\big\}\cdot\bigg|\bbE\bigg[\prod_{ij\in\alpha}\frac{\bA_{ij}-p}{\sqrt{p(1-p)}}\bigg]\bigg| \\
&= \bbP\big\{\alpha\setminus\tbinom{\cS}{2}=\beta\big\}\cdot\left(\frac{p}{1-p}\right)^{|\alpha|/2}\bigg|\bbE\bigg[\prod_{ij\in\alpha}\inner{X_i,X_j}\bigg]\bigg| \\
&\leq \bbP\big\{\alpha\setminus\tbinom{\cS}{2}=\beta\big\} \leq r^{|V(\alpha\setminus\beta)|} \,.
\end{align*}
Finally we bound $|\bbE[\bar{\bA}^\alpha\cdot\theta]|$. Let $X_1,\dots,X_{|V(\alpha)|}$ be i.i.d.\ uniformly distributed on the sphere $\bbS^{d-1}$. Since $\bbE\left[\bar{\bA}^\alpha\,|\,\cS\right]\cdot\theta\neq0$ only if $V(\alpha)\cup\{1\}\subseteq\cS$, we have
\begin{align*}
\bbE[\bar{\bA}^\alpha\cdot\theta] &= \bbE_{\cS}\left[\bbE\left[\bar{\bA}^\alpha\,|\,\cS\right]\cdot\theta\right] \\
&= \bbP\{V(\alpha)\cup\{1\}\subseteq\cS\}\cdot\bbE\left[\bar{\bA}^\alpha\,|\,V(\alpha)\subseteq\cS\right] \leq r^{|V(\alpha)\cup\{1\}|} \leq r^{|V(\alpha)|} \,,
\end{align*}
where the first inequality holds via the same calculation as for $|M_{\alpha\beta}|$. The above bounds on $M_{\alpha\alpha}$, $|M_{\alpha\beta}|$, and $|\bbE[\bar{\bA}^\alpha\cdot\theta]|$ match those in \cite{schramm2022computational}, so indeed Lemma H.4 holds in our setting, completing the proof.
\end{proof}


\section{Information-theoretic upper bound}\label{sec:ITUpperBound}

We prove \Cref{thm:ITUpperBound} in this section. 
For clarity, let $\cS^\ast$ (instead of $\cS$) denote the vertex set of the planted subgraph in $\bA$ throughout the proof. If $k\geq cn$ for a constant $c>0$ and we take $\epsilon\in(0,c^2)$ then $|\cS^\ast\cap\hat{\cS}|\geq\epsilon k$ with high probability (recall $\hat{\cS}$ is defined as a random $k$-subset of $[n]$ when $k=\Theta(n)$). Thus in the remainder we may assume $k=o(n)$, where $\hat{\cS}$ is defined by \eqref{eqn:Shat-kon-dfn}. 
For all subsets $A,B,C\subseteq[n]$ define
$$\tau(A,B,C):=\sum_{i\in A,j\in B,l\in C}\bar{\bA}_{ij}\bar{\bA}_{jl}\bar{\bA}_{li}\,$$
and write $\tau(A) \equiv \tau(A,A,A)$. Moreover, define
\[\arraycolsep=2mm
\begin{array}{lll}
T_1(A):=\tau(A,\cS^\ast\setminus A,\cS^\ast\setminus A) \,, & T_2(A):=\tau(A,A,\cS^\ast\setminus A) \,, & T_3(A):=\tau(A) \,, \\[8pt]
T_4(A,B):=\tau(B,\cS^\ast\setminus A,\cS^\ast\setminus A) \,, & T_5(A,B):=\tau(B,B,\cS^\ast\setminus A) \,, & T_6(B):=\tau(B) \,.
\end{array}\]
At times we will write $T_1,\dots,T_6$ without their respective arguments to ease notation. Define
$$\Delta(A,B):=3T_1(A)+3T_2(A)+T_3(A)-3T_4(A,B)-3T_5(A,B)-T_6(B)\,.$$
By definition of $\Delta$, if $\cS\in\binom{[n]}{k}$, $A=\cS^\ast\setminus\cS$, and $B=\cS\setminus\cS^\ast$, then $\cS^\ast\setminus A = \cS^\ast \cap \cS$ and 
$$\Delta(A,B)=\tau(\cS^\ast)-\tau(\cS).$$

Let $\epsilon>0$ be a sufficiently small constant. The optimality of $\hat{\cS}$ implies $\Delta(\cS^\ast\setminus\hat{\cS},\hat{\cS}\setminus\cS^\ast)$ is nonpositive, hence
\begin{align*}
\bbP\{\abs{\cS^\ast\cap\hat{\cS}}\leq\epsilon k\} &\leq \bbP\left\{\tau(\cS)\geq\tau(\cS^\ast)\text{ for some }\cS\in\binom{[n]}{k}\text{ with }\abs{\cS^\ast\cap\cS}\leq\epsilon k\right\} \\
&= \bbP\left\{\Delta(A,B)\leq0\text{ for some }A\in\binom{\cS^\ast}{\geq(1-\epsilon)k},\,B\in\binom{[n]\setminus\cS^\ast}{|A|}\right\} \\
&\leq \sum_{A\in\binom{\cS^\ast}{\geq(1-\epsilon)k}}\bbP\left\{\Delta(A,B)\leq0\text{ for some }B\in\binom{[n]\setminus\cS^\ast}{\abs{A}}\right\} \,,
\end{align*}
and let us write $P_A$ to denote the summand on the right-hand side. Let $s:=p^3k^3/2d^2$, and for all $A$ define the event
$$E_A:=\{T_1(A)\leq-s/12\text{ or }T_2(A)\leq-s/12\text{ or }T_3(A)\leq s\}\,.$$
We then have
\begin{align*}
P_A &\leq \bbP\{E_A\} + \bbP\left\{\Delta(A,B)\leq0\text{ for some }B\in\binom{[n]\setminus\cS^\ast}{\abs{A}} \text{ and $E_A^c$ holds} \right\} \\
&\leq \sum_{i=1}^2\bbP\left\{T_i\leq-s/12\right\} + \bbP\{T_3\leq s\} \\
&\qquad + \bbP\left\{3T_4+3T_5+T_6\geq s/2\text{ for some }B\in\binom{[n]\setminus\cS^\ast}{\abs{A}}\right\} \\
&\leq \sum_{i=1}^2\bbP\left\{T_i\leq-s/12\right\} + \bbP\{T_3\leq s\} + \sum_{i=4}^6\bbP\left\{T_i\geq s/14\text{ for some }B\in\binom{[n]\setminus\cS^\ast}{\abs{A}}\right\} \,.
\end{align*}
Summarizing, and applying another union bound, we write
\begin{equation}
\begin{aligned}
\bbP\{\abs{\cS^\ast\cap\hat{\cS}}\leq\epsilon k\} &\leq \sum_{A\in\binom{\cS^\ast}{\geq(1-\epsilon)k}}\left(\bbP\left\{T_1\leq-s/12\right\} + \bbP\{T_2\leq-s/12\} + \bbP\{T_3\leq s\}\right) \\
&\hspace{2cm}+ \sum_{i=4}^6\sum_{A\in\binom{\cS^\ast}{\geq(1-\epsilon)k}}\sum_{B\in\binom{[n]\setminus\cS^\ast}{|A|}}\bbP\left\{T_i\geq s/14\right\} \,.
\end{aligned}\label{eqn:MainTermSstarShat}
\end{equation}
If $H(x)=-x\log_2x-(1-x)\log_2(1-x)$ is the binary entropy then we have
\begin{equation}
\binom{k}{\geq(1-\epsilon)k}\leq2^{H(\epsilon)k}\qquad\text{and}\qquad\binom{k}{\geq(1-\epsilon)k}\binom{n-k}{k}\leq2^{H(\epsilon)k}e^{k\log(en/k)}\,.\label{eqn:binomupperbds}
\end{equation}
To complete the proof, it suffices to show that for all $A\in\binom{\cS^\ast}{\geq(1-\epsilon)k}$ and $B\in\binom{[n]\setminus\cS^\ast}{|A|}$, each of the summands in \eqref{eqn:MainTermSstarShat} is $o(\cdot)$ of the reciprocal of the corresponding term in \eqref{eqn:binomupperbds}.

\noindent{\textbf{Tail bounds for $T_1$, $T_2$, and $T_3$.}} We will prove a tail bound for $T_1$ in two steps: first we control the conditional expectation $M:=\bbE[T_1\,|\,\bX]$; then we control $T_1$ given $M$. The function
$$f(x_i\,;\,i\in\cS^\ast):=p^3\sum_{i\in A}\,\sum_{\substack{j,l\in\cS^\ast\setminus A\\j\neq l}}\inner{x_i,x_j}\inner{x_j,x_l}\inner{x_l,x_i}$$
satisfies $M=f(X_i\,;\,i\in\cS^\ast)$, and $f$ satisfies the bounded differences property with parameters $c_i=2p^3\abs{\cS^\ast\setminus A}^2$ for $i\in A$, and $c_i:=2p^3|A|\cdot\abs{\cS^\ast\setminus A}$ for $i\in\cS^\ast\setminus A$. If we let $\lambda:=\bbE M+s/24$ then by the bounded differences inequality (Theorem 6.2 in \cite{boucheron2013concentration}), 
\begin{equation}
\bbP\left\{M\leq \bbE M-\lambda\right\}\leq\exp\left(-\frac{k^5}{C d^4|A|\cdot|\cS^\ast\setminus A|^3}\right) \leq \exp\left(-\frac{k}{C \epsilon^3d^4}\right)  \label{eqn:t1bbddiffM}
\end{equation}
for an absolute constant $C>0$ (that may vary between lines), where we used that $|A|\leq k$ and $|\cS^\ast\setminus A|\leq\epsilon k$. Let $E$ be the event on the left-hand side of \eqref{eqn:t1bbddiffM}. Using Theorem 2.1 from \cite{janson2004large} and redefining $\lambda:=M+s/24$,
\begin{equation}
\bbP\left\{T_1\leq M-\lambda \text{ and $E^c$ holds}\right\}\leq\exp\left(-\frac{s^2}{C|A|^2|\cS^\ast\setminus A|^2}\right) \leq \exp\left(-\frac{p^6k^2}{C\epsilon^2d^4}\right)\,. \label{eqn:t1mlambdajanson}
\end{equation}
To apply Theorem 2.1 from \cite{janson2004large}, we simply used that the maximum degree of the dependency graph of the variables $\{\bar{\bA}_{ij}\bar{\bA}_{jl}\bar{\bA}_{li}\}_{i\in A,\,j,l\in\cS^\ast\setminus A}$ (conditioned on $\bX$) is at most $3|A|-2$ (an edge joins two variables in the dependency graph if and only if they are not independent). Combining \eqref{eqn:t1bbddiffM} and \eqref{eqn:t1mlambdajanson}, we have 
\begin{align*}
    \bbP\{T_1\leq-s/12\} &\leq \bbP\{E\}+\bbP\left\{T_1\leq M-\lambda \text{ and } E^c \text{ holds}\right\} \\
    &\leq \exp\left(-\frac{k}{C\epsilon^3d^4}\right)+\exp\left(-\frac{p^6k^2}{C\epsilon^2d^4}\right)\,.
\end{align*}
Using precisely the same argument as for $T_1$, we calculate that
$$\bbP\{T_2\leq-s/12\}\leq\exp\left(-\frac{k}{C\epsilon d^4}\right)+\exp\left(-\frac{p^6k^2}{C\epsilon d^4}\right).$$
For the term $T_3$, the argument only needs to be modified slightly. Note that 
\[
    \bbE M = p^3 \sum_{i,j,l \in A} 1/d^2 \ge 0.99 p^3 k^3/d^2 = 1.98 s.
\]
In this case, the function
$$f(x_i\,;\,i\in\cS^\ast):=p^3\sum_{\substack{i,j,l\in A\\\text{distinct}}}\inner{x_i,x_j}\inner{x_j,x_l}\inner{x_l,x_i}$$
satisfies the bounded differences property with parameters $c_i=2p^3|A|^2$, so the variance term in \cite[Theorem~6.2]{boucheron2013concentration} is $p^6|A|^5$. Therefore, choosing $\lambda = 0.4 s$ in the above argument yields
$$\bbP\{T_3\leq s\}\leq\exp\left(-\frac{k}{Cd^4}\right)+\exp\left(-\frac{p^6k^2}{Cd^4}\right)\,.$$
If $\epsilon>0$ is taken to sufficiently small depending on $C$ and $d$, then above tail bounds are all $o(2^{-H(\epsilon)k})$.

\noindent{\textbf{Tail bounds for $T_4$, $T_5$, and $T_6$.}} For $T_4$ and $T_5$, we may assume $\cS^\ast\setminus A\neq\emptyset$ (otherwise these terms are 0 trivially). We first bound $T_4$. Notice that $\bbE[T_4\,|\,\bX]=0$ with probability 1, since for $i\in B$ and $j\in\cS^\ast\setminus A$, the random variable $\bA_{ij}$ is centered and independent of all else. To apply Theorem 2.1 from \cite{janson2004large}, note that the maximum degree of the dependency graph of the random variables $\{\bar{\bA}_{ij}\bar{\bA}_{jl}\bar{\bA}_{li}\}_{i\in B,j,l\in\cS^\ast\setminus A}$ (conditional on $\bX$) is at most $3k$. We thus have
\begin{align*}
\bbP\{T_4\geq s/6\,|\,\bX\} &\leq \exp\left(-2\cdot\frac{(s/6)^2}{3k\cdot\abs{B}\cdot\abs{\cS^\ast\setminus A}^2}\right) \\
&\leq \exp\left(-\frac{p^6k^2}{216\epsilon^2d^4}\right) = o(2^{-H(\epsilon)k}e^{-k\log(en/k)}) \,,
\end{align*}
where we used that $|B|\leq k$ and $|\cS^\ast\setminus A|\leq\epsilon k$, and in the last step we used that $k\geq C\log n$ for a large enough constant $C=C(p,d)>0$. Taking an expectation with respect to $\bX$ completes the proof.

Analogous tail bounds for $T_5$ and $T_6$ are proved in exactly the same manner as $T_4$. For $T_5$, the maximum degree of the dependency graph of the variables $\{\bar{\bA}_{ij}\bar{\bA}_{jl}\bar{\bA}_{li}\}_{i,j\in B,l\in\cS^\ast\setminus A}$ (conditional on $\bX$) is again at most $3k$. For $T_6$, we need not even condition on $\bX$ since $T_6$ is simply the signed triangle count of an \erdosrenyi{} random graph.


\section{Probability and linear algebra tools}

The following theorem is useful for multiple proofs in this work.

\begin{theorem}\label{thm:BvH}{\normalfont(Corollary~3.12 and Remark~3.13 in \cite{bandeira2016sharp})}
Let $X$ be an $n\times n$ symmetric matrix whose entries $X_{ij}$ are centered random variables, independent up to symmetry. Define the quantities
$$\sigma:=\max_{i=1,\dots,n}\sqrt{\sum_{j=1}^n\bbE X_{ij}^2}\,,\hspace{10mm}\sigma_\ast:=\max_{i,j=1,\dots,n}\norm{X_{ij}}_\infty\,,$$
where we write $X_{ij}^2=(X_{ij})^2$. Then for all $0<\epsilon\leq\frac{1}{2}$ there exists a universal constant $c_\epsilon$ such that for every $t\geq0$,
$$\bbP\{\norm{X}\geq(1+\epsilon)2\sqrt{2}\sigma+t\}\leq n\exp\left(-\frac{t^2}{c_\epsilon\sigma_\ast^2}\right)\,.$$
\end{theorem}

The following is a useful result about the spectral norm of a Hadamard product of matrices.

\begin{lemma}[(3.7.9) in \cite{horn1990hadamard}] \label{lem:hadamard-spectral-norm}
For matrices $B \in \R^{n \times n}$, $X, Y \in \R^{m \times n}$, and $A = X^\top Y$, we have
$$
\|A \circ B\| \le \sqrt{\|\bI_n \circ (X^\top X)\| \cdot \|\bI_n \circ (Y^\top Y)\|} \cdot \|B\| .
$$
In particular, if $X = Y$, then 
$$
\|A \circ B\| \le \max_{i \in [n]} A_{ii} \cdot \|B\| .
$$
\end{lemma}
\end{document}